%% file: main.tex
\documentclass[twoside,11pt]{article}

\usepackage[margin=3.2cm]{geometry}
\usepackage[utf8]{inputenc}
\usepackage{booktabs} 
\usepackage{hyperref}
\usepackage{amsmath,xfrac,stackengine,amsthm}
\usepackage{graphicx}
\usepackage{amssymb}
\usepackage{dsfont} 
\usepackage[font=small,skip=2pt]{caption}
\usepackage[english]{babel}
\usepackage{enumitem}
\usepackage{subfig,float}
\usepackage{algorithm} 
\usepackage[noend]{algpseudocode}
\usepackage{xfrac}
\usepackage{framed}
\usepackage[table,xcdraw]{xcolor}
\usepackage{multirow}

\usepackage[authoryear,round]{natbib}
\usepackage[percent]{overpic}

\renewcommand{\paragraph}[1]{\vspace{2mm}\noindent\textbf{#1}}
\newcommand{\diag}[1]{\text{diag}(#1)}
\newcommand{\trace}[1]{\text{tr}(#1)}

\newcommand{\spanning}[1]{\text{span}(#1)}
\newcommand{\image}[1]{\text{im}(#1)}
\newcommand{\dimension}[1]{\text{dim}(#1)}
\newcommand{\norm}[1]{\left\lVert#1\right\rVert}

\newcommand{\defequal}{\overset{\Delta}{=}}
\newcommand{\1}{{1}}

\DeclareMathOperator*{\argmin}{arg\,min}

\newcommand{\E}[1]{\mathbf{E}\hspace{-1px}\left[#1\right]}

\newcommand{\Rbb}{\mathbb{R}}

\renewcommand{\deg}{d} 

\newcommand{\sintheta}[2]{\sin{\hspace{0.0px}\Theta\hspace{-0.0px}\big( #1, #2 \big)}}

\renewcommand{\c}[1]{{#1}_{c}} 
\newcommand{\p}[1]{\widetilde{#1}}

\newcommand{\C}{\mathcal{C}}  
\newcommand{\F}{\mathcal{F}}  
\renewcommand{\P}{\mathcal{P}}  
\newcommand{\V}{\mathcal{V}}  
\renewcommand{\E}{\mathcal{E}}  
\renewcommand{\S}{\mathcal{S}}  

\newcommand{\kmeans}[3]{\mathcal{F}_{#1}\hspace{-0.5px}(#2,#3)}

\newtheoremstyle{style}
{6pt} 
{4pt} 
{\itshape} 
{} 
{\bfseries} 
{.} 
{.5em} 
{} 

\theoremstyle{style}
\newtheorem{theorem}{Theorem}[section]
\newtheorem{proposition}{Proposition}[section]
\newtheorem{definition}[theorem]{Definition}
\newtheorem{lemma}[theorem]{Lemma}
\newtheorem{property}[theorem]{Property}
\newtheorem{corollary}[theorem]{Corollary}

\newif\ifhideproofs
\ifhideproofs
  \usepackage{environ}
  \NewEnviron{hide}{}

\fi

\makeatletter
\newcommand\framename{Scheme}
\newcounter{framecnt}
\newcommand{\TitleFrame}[2]{%
    \fboxrule=\FrameRule
    \fboxsep=\FrameSep
    \vbox{\nobreak \vskip 0.2\FrameSep
        \rlap{\strut#1}\nobreak\nointerlineskip
        \vskip 0.05\FrameSep
        \noindent\fbox{#2}}
    }
\newenvironment{titledframe}[2][\FrameFirst@Lab\ (cont.)]{%
    \refstepcounter{framecnt}%
    \def\FrameFirst@Lab{\textbf{\framename\ \theframecnt:\ #2}}%
    \def\FrameCont@Lab{\textbf{#1}}%
    \def\FrameCommand##1{%
        \TitleFrame{\FrameFirst@Lab}{##1}}%
    \def\FirstFrameCommand##1{%
        \TitleFrame{\FrameFirst@Lab}{##1}}%
    \def\MidFrameCommand##1{%
        \TitleFrame{\FrameCont@Lab}{##1}}%
    \def\LastFrameCommand##1{%
        \TitleFrame{\FrameCont@Lab}{##1}}%
    \MakeFramed{\hsize\textwidth
    \advance\hsize -2\FrameRule
    \advance\hsize -2\FrameSep
    \FrameRestore}}%
   {\endMakeFramed}
\makeatother

\begin{document} 

\title{Graph reduction with spectral and cut guarantees}

\author{Andreas Loukas\\
       \'{E}cole Polytechnique F\'{e}d\'{e}rale Lausanne, \\
       Switzerland
}

\date{} 
\maketitle

\begin{abstract}
Can one reduce the size of a graph without significantly altering its basic properties? The graph reduction problem is hereby approached from the perspective of \emph{restricted spectral approximation}, a modification of the spectral similarity measure used for graph sparsification. This choice is motivated by the observation that restricted approximation carries strong spectral and cut guarantees, and that it implies approximation results for unsupervised learning problems relying on spectral embeddings. 

The paper then focuses on {coarsening}---the most common type of graph reduction. 
Sufficient conditions are derived for a small graph to approximate a larger one in the sense of restricted similarity. These findings give rise to nearly-linear algorithms that, compared to both standard and advanced graph reduction methods, find coarse graphs of improved quality, often by a large margin, without sacrificing speed.
\end{abstract}

\section{Introduction}

As graphs grow in size, it becomes pertinent to look for generic ways of simplifying their structure while preserving key properties. 
Simplified graph representations find profound use in the design of approximation algorithms, can facilitate storage and retrieval, and ultimately ease graph data analysis by separating overall trends from details.

There are two main ways to simplify graphs. 
First, one may reduce the number of edges, a technique commonly referred to as \emph{graph sparsification}.  
In a series of works, it has been shown that it is possible to find sparse graphs that approximate all pairwise distances~\citep{peleg1989graph}, every cut~\citep{karger1999random}, or every eigenvalue~\citep{spielman2011spectral}---respectively referred to as {spanners}, {cut sparsifiers} and {spectral sparsifiers}.
Spectral sparsification techniques in particular can yield computational benefits whenever the number of edges is the main bottleneck~\citep{Batson:2013:SSG:2492007.2492029}. Indeed, they form a fundamental component of nearly-linear time algorithms for linear systems involving symmetric diagonally dominant matrices~\citep{5671167,spielman2011graph}, and have found application to machine learning problems involving graph-structured data~\citep{calandriello2018improved}. 

Alternatively, one may seek to reduce directly the size of the graph, i.e., the number of its vertices $N$, by some form of vertex selection or re-combination scheme followed by re-wiring. 
This idea can be traced back to the mutligrid literature, that targets the acceleration of finite-element methods using cycles of multi-level coarsening, lifting and refinement. 
After being generalized to graphs, reduction methods have become pervasive in computer science and form a key element of modern graph processing pipelines, especially with regards to graph partitioning~\citep{hendrickson1995multi,karypis1998fast,kushnir2006fast,dhillon2007weighted,wang2014partition} and graph visualization~\citep{koren2002fast,hu2005efficient,walshaw2006multilevel}. 
In machine learning, reduction methods are used to create multi-scale representations of graph-structured data~\citep{lafon2006diffusion,gavish2010multiscale,shuman2016multiscale} and as a layer of graph convolutional neural networks~\citep{bruna2014spectral,defferrard2016convolutional,7974879,simonovsky2017dynamic,ardizzone2018analyzing}. In addition, being shown to solve linear systems in (empirically) linear time~\citep{koutis2011combinatorial,livne2012lean} as well as to approximate the Fiedler vector~\citep{urschel2014cascadic,gandhi2016improvement}, reduction methods have been considered as a way of accelerating graph-regularized problems~\citep{hirani2015graph,colley2017algebraic}. 
Some of their main benefits are the ability to deal with sparse graphs --graphs with at most $O(N\log{N})$ edges-- and to accelerate algorithms whose complexity depends on the number of vertices as well as edges.

Yet, in contrast to graph sparsification, there has been only circumstantial theory supporting graph reduction~\citep{moitra2011vertex,dorfler2013kron,loukas2018spectrally}. The lack of a concrete understanding of how different reduction choices affect fundamental graph properties is an issue: the significant majority of reduction algorithms in modern graph processing and machine learning pipelines have been designed based on intuition and possess no rigorous justification or provable guarantees.

\paragraph{A new perspective.}
My starting point in this work is \emph{spectral similarity}---a measure that has been proven useful in sparsification for determining how well a graph approximates another one. To render spectral similarity applicable to graphs of different sizes, I generalize it and restrict it over a subspace of size that is at most equal to the size of the reduced graph. I refer to the resulting definition as \emph{restricted spectral approximation}\footnote{Though similarly named, the definition of \emph{restricted spectral similarity} previously proposed by~\citep{loukas2018spectrally} concerns a set of vectors (rather than subspaces) and is significantly weaker than the one examined here.} (or restricted approximation for short). 
Despite being a statement about subspaces, restricted similarity has important consequences. 
It is shown that when the subspace in question is a principal eigenspace (this is a data agnostic choice where one wants to preserve the global graph structure), the eigenvalues and eigenspaces of the reduced graph approximate those of the original large graph. It is then a corollary that (\emph{i}) if the large graph has a good cut so does the smaller one; and (\emph{ii}) that unsupervised learning algorithms that utilize spectral embeddings, such as spectral clustering~\citep{von2007tutorial} and Laplacian eigenmaps~\citep{belkin2003laplacian}, can also work well when run on the smaller graph and their solution is lifted.

The analysis then focuses on \emph{graph coarsening}---a popular type of reduction where, in each level, reduced vertices are formed by contracting disjoint sets of connected vertices (each such set is called a \emph{contraction set}).
I derive sufficient conditions for a small coarse graph to approximate a larger graph in the sense of restricted spectral approximation. Crucially, this result holds for any number of levels and 
is independent of how the subspace is chosen. Though the derived bound is global, a decoupling argument renders it locally separable over levels and contraction sets, facilitating computation. The final bound can be interpreted as measuring the \emph{local variation} over each contraction set, as it involves the maximum variation of vectors supported on each induced subgraph.

These findings give rise to greedy nearly-linear time algorithms for graph coarsening, that I refer to as \emph{local variation algorithms}. %
Each such algorithm starts from a predefined family of candidate contraction sets. Even though any connected set of vertices may form a valid candidate set, I opt for small well-connected sets, formed for example by pairs of adjacent vertices or neighborhoods.
The algorithm then greedily\footnote{Even after decoupling, the problem of candidate set selection is not only NP-hard but also cannot be approximated to a constant factor in polynomial time (by reduction to the maximum-weight independent set problem). For the specific case of edge-based families, where one candidate set is constructed for each pair of adjacent vertices, the greedy iterative contraction can be substituted by more sophisticated procedures accompanied by improved guarantees.} contracts those sets whose local variation is the smallest. 
Depending on how the candidate family is constructed, the proposed algorithms obtain different solutions, trading off computational complexity for reduction.

\paragraph{Theoretical and practical implications.} Despite not providing a definitive answer on how much one may gain (in terms of reduction) for a given error, the analysis improves and generalizes upon previous works in a number of ways: 

\begin{itemize}
\item Instead of directly focusing on specific constructions, a general graph reduction scheme is studied featuring coarsening as a special case. As a consequence, the implications of restricted similarity are proven in a fairly general setting where specifics of the reduction (such as the type of graph representation and the reduction matrices involved) are abstracted. 

\item Contrary to previous results on the analysis of coarsening~\citep{loukas2018spectrally}, the analysis holds for multiple levels of reduction. Given that the majority of coarsening methods reduce the number of vertices by a constant factor at each level, a multi-level approach is necessary to achieve significant reduction. Along that line, the analysis also brings an intuitive insight: rather than taking the common approach of approximating at each level the graph produced by the previous level, one should strive to preserve the properties of the original graph at every level.

\item The proposed local variation algorithms are not heuristically designed, but greedily optimize (an upper bound of) the restricted spectral approximation objective. Despite the breadth of the literature that utilizes some form of graph reduction and coarsening, the overwhelming majority of known methods are heuristics---see for instance~\citep{safro2015advanced}. A notable exception is Kron reduction~\citep{dorfler2013kron}, an elegant method that aims to preserve the effective resistance distance. Compared to Kron reduction, the graph coarsening methods proposed here are accompanied by significantly stronger spectral guarantees (i.e., beyond interlacing), do not sacrifice the sparsity of the graph, and can ultimately be more scalable as they do not rely on the Schur complement of the Laplacian matrix. 
\end{itemize}

To demonstrate the practical benefits of local variation methods, the analysis is complemented with numerical results on representative graphs ranging from scale-free graphs to meshes and road networks. Compared to both standard~\citep{karypis1998fast} and advanced reduction methods~\citep{ron2011relaxation,livne2012lean,shuman2016multiscale}, the proposed methods yield small graphs of improved spectral quality, often by a large margin, without being much slower than naive heavy-edge matching. A case in point: when examining how close are the principal eigenvalues of the coarse and original graph for a reduction of 70\%, local variation methods attain on average 2.6$\times$ smaller error; this gain becomes 3.9$\times$ if one does not include Kron reduction in the comparison.

\section{Graph reduction and coarsening}
\label{sec:model}

The following section introduces graph reduction.
The exposition starts by considering a general reduction scheme. It is then shown how graph coarsening arises naturally if one additionally imposes requirements w.r.t. the interpretability of reduced variables.

\subsection{Graph reduction}

Consider a positive semidefinite (PSD) matrix $L \in \Rbb^{N\times N}$ whose sparsity structure captures the connectivity structure of a connected weighted symmetric graph $G = (\V, \E, W)$ of $N = |\V|$ vertices and $M = |\E|$ edges. 
In other words, $L(i,j) \neq 0$ only if $e_{ij}$ is a valid edge. 
Moreover, let $x$ be an arbitrary vector of size $N$.

I study the following generic reduction scheme: 

\begin{titledframe}{Graph reduction}\label{frm:matrix_reduction}
Commence by setting $  L_{0}=L$ and $x_{0}=x $
and proceed according to the following two recursive equations:
\begin{align}
    L_{\ell} = P_\ell^{\mp} L_{\ell-1} P_\ell^+ \quad \text{and} \quad x_\ell = P_{\ell} \, x_{\ell-1}, \notag 
\end{align}
where $P_\ell \in \Rbb^{N_\ell \times N_{\ell-1}}$ are matrices with more columns than rows, $\ell = 1, 2, \ldots, c$ is the level of the reduction, symbol ${\mp}$ denotes the transposed pseudoinverse, and $N_\ell$ is the dimensionality at level $\ell$ such that $N_0 = N$ and $N_{c} = n \ll N$.

\vspace{1.5mm}
Vector $x_c$ is lifted back to $\Rbb^N$ by recursion $\p{x}_{\ell-1} = P_{\ell}^+ \p{x}_{\ell}$, where $\p{x}_{c} = x_{c}.$
\end{titledframe}

Graph reduction thus involves a sequence of $c+1$ graphs  
\begin{align}
    G = G_0 = (\V_0, \E_0, W_0) \quad G_1 = (\V_1, \E_1, W_1) \quad \cdots \quad G_{c} = (\V_{c}, \E_{c}, W_{c})
\end{align}
of decreasing size
$N = N_0  >  N_{1} > \cdots > N_{c} = n$, where the sparsity structure of $L_{\ell}$ matches that of graph $G_{\ell}$, and each vertex of $G_{\ell}$ represents one of more vertices of $G_{\ell-1}$.

The {multi-level} design allows us to achieve high \emph{dimensionality reduction ratio} $$r = 1 - \frac{n}{N},$$ even when at each level the {dimensionality reduction ratio} $r_\ell = 1 - \frac{N_\ell}{N_{\ell-1}}$ is small. For instance, supposing that  $r_\ell \geq \varrho$ for each $\ell$, then $c = O( \log(n/N)/ \log(1-\varrho)) $ levels suffice to reduce the dimension to $n$. 

One may express the reduced quantities in a more compact form:  
\begin{align} 
    x_{c} &=  P x, \quad L_{c} = P^{\mp} L P^+  \quad \text{and} \quad \p{x} = \Pi x,
    %
    %
\end{align}
where $P = P_{c} \cdots P_{1}$, $P^+ = P_1^+ \cdots P_{c}^+$ and $\Pi = P^+ P$. For convenience, I drop zero indices and refer to a lifted vector as $\p{x} (= \p{x}_0)$.

The rational of this scheme is that vector $\p{x}$ should be the best approximation of $x$ given $P$ in an $\ell_2$-sense, which is a consequence of the following property:
\begin{property}
$\Pi$ is a projection matrix.  
\label{property:Pi_projection}
\end{property}
On the other hand, matrix $L$ is reduced such that $x_c^\top L_c x_c = \p{x}^\top L \p{x}$. 

Though introduced here for the reduction of sparse PSD matrices representing the similarity structure of a graph, Scheme~\ref{frm:matrix_reduction} can also be applied to any PSD matrix $L$. In fact, this and similar reduction schemes belong to the class of Nystr\"om methods and, to the extend of my knowledge, they were first studied in the context of approximate low-rank matrix approximation~\citep{halko2011finding,wang2013improving}. Despite the common starting point, interpreting $L$ and $L_c$ as sparse similarity matrices, as it is done here, incorporates a graph-theoretic twist to reduction that distinguishes from previous methods\footnote{To achieve low-rank approximation, matrix $P$ is usually built by sampling columns of $L$.}: the constructions that we will study are eventually more scalable and interpretable as they maintain the graph structure of $L$ after reduction. Obtaining guarantees is also significantly more challenging in this setting, as the involved problems end up being combinatorial in nature.

\subsection{Properties of reduced graphs} 

Even in this general context where $P$ is an arbitrary $n \times N$ matrix, certain handy properties can be proven about the relation between $L_c$ and $L$. 

To begin with, it is simple to see that the set of positive semidefinite matrices is closed under reduction.
\begin{property}
If $L$ is PSD, then so is $L_c$. 
\end{property}
The proof is elementary: if $L$ is PSD then there exists matrix $S$ such that $L = S^\top S$, implying that $L_c = P^\mp L P^+$ can also be written as $L_c = S_c^\top S_c$ if one sets $S_c = S P^+$.

I further consider the spectrum of the two matrices. Sort the eigenvalues of $L$ as $\lambda_1 \leq \lambda_2 \leq \ldots \leq \lambda_N$  and denote by $\p{\lambda}_k$ the $k$-th largest eigenvalue of $L_{c}$ and $\p{u}_k$ the associated eigenvector. 

It turns out that the eigenvalues $\p{\lambda}$ and $\lambda$ are interlaced. 

\begin{theorem}
For any $P$ with full-row rank and $k = 1, \ldots, n$, we have 
$$ \gamma_1 \, \lambda_k  \leq \tilde{\lambda}_k \leq \gamma_2 \,  \lambda_{k+N-n}$$ 
with $\gamma_1 =\lambda_1( (P P^\top)^{-1})$ and $\gamma_2 = \lambda_n((P P^\top)^{-1})$, respectively the smallest and largest eigenvalue of $(P P^\top)^{-1}$.
\label{theorem:interlacing}
\end{theorem}
The above result is a generalization of the Cauchy interlacing theorem for the case that $P P^\top \neq I$. 
It also resembles the interlacing inequalities known for the normalized Laplacian (where the re-normalization is obtained by construction).
\citet{chen2004interlacing} showed in Theorem 2.7 of their paper that after contracting $N-n$ edges 
$
\lambda_{k - N + n} \leq \lambda_k \leq \lambda_{k + N - n}
$
for $k = 1, 2, \ldots, n$ and with $\lambda_{\ell} = 0$ when $\ell \leq 0$, resembling the upper bound above.  The lower bound is akin to that given in \citep[Lemma 1.15]{chung1997spectral}, again for the normalized Laplacian. Also notably, the inequalities are similar to those known for Kron reduction~\citep[Lemma 3.6]{dorfler2013kron}.

Theorem~\ref{theorem:interlacing} is particularly pessimistic as it has to hold for {every} possible $P$ and $L$. Much stronger results will be obtained later on by restricting the attention to constructions that satisfy additional properties (see Theorem~\ref{theorem:eigenvalues}).

One can also say something about the eigenvectors of $L_c$.
\begin{property}
For every vector for which $x = \Pi x$, one has
\begin{align}
    \c{x}^\top \c{L} \c{x} = x^\top \Pi L \Pi x = x^\top L x \quad \text{and} \quad \p{x} = \Pi x = x. \notag
\end{align} 
\end{property}
In other words, reduction maintains the action of $L$ of every vector that lies in the range of $\Pi$. Naturally, after lifting the eigenvectors of $\c{L}$ are included in this class. 

\subsection{Coarsening as a type of graph reduction} 
\label{subsec:graph_coarsening}

Coarsening is a type of graph reduction abiding to a set of constraints that render the graph transformation interpretable. 
More precisely, in coarsening one selects for each level $\ell$ a surjective (i.e., many-to-one) map $\varphi_{\ell}: \V_{\ell-1} \rightarrow \V_{\ell}$ between the original vertex set $\V_{\ell-1}$ and the smaller vertex set $\V_{\ell}$. 
I refer to the set of vertices $\V_{\ell-1}^{(r)} \subseteq \V_{\ell-1}$ mapped onto the same vertex $v_r'$ of $\V_{\ell}$ as a \textit{contraction set}: $$\V_{\ell-1}^{(r)} = \{v \in \V_{\ell-1} : \varphi_{\ell}(v) = v'_r \}$$ 
For a graphical depiction of contraction sets, see Figure~\ref{fig:example}.
I also constrain $\varphi_{\ell}$ slightly by requiring that the subgraph of $G_{\ell-1}$ induced by each contraction set $\V_{\ell-1}^{(r)}$ is connected. 

It is easy to deduce that {contraction sets} induce a partitioning of $\V_{\ell-1}$ into $N_{\ell}$ subgraphs, each corresponding to a single vertex of $\V_{\ell}$.
Every reduced variable thus corresponds to a small set of adjacent vertices in the original graph and coarsening basically amounts to a scaling operation. An appropriately constructed coarse graph aims to capture the global problem structure, whereas neglected details can be recovered in a local refinement phase.

Coarsening can be placed in the context of Scheme~\ref{frm:matrix_reduction} by restricting each $P_{\ell}$ to lie in the family of coarsening matrices, defined next:
\begin{definition}[Coarsening matrix]
Matrix $P_{\ell} \in \Rbb^{N_{\ell} \times N_{\ell-1}}$ is a coarsening matrix w.r.t. graph $G_{\ell-1}$ if and only if it satisfies the following two conditions:
\begin{itemize}
    \item[a.] It is a \textit{surjective mapping} of the vertex set, meaning that if $P_{\ell}(r,i) \neq 0$ then $P_{\ell}(r',i) = 0$ for every $r' \neq r$. 
    \item[b.] It is \textit{locality preserving}, equivalently, the subgraph of $G_{\ell-1}$ induced by the non-zero entries of $P_{\ell}(r,:)$ is connected for each $r$.
\end{itemize}
\label{def:coarsening_matrix}
\end{definition}

An interesting consequence of this definition is that, in contrast to graph reduction, with coarsening matrices the expensive pseudo-inverse computation can be substituted by simple transposition and re-scaling:
\begin{proposition}[Easy inversion]
The pseudo-inverse of a coarsening matrix $P_{\ell}$ is given by $P_\ell^+ =  P_\ell^\top D_{\ell}^{-2}$, where $D_{\ell}$ is the diagonal matrix with $D_\ell(r,r) = \| P_\ell(r,:)\|_2$. 
\label{proposition:pseudo-inverse}
\end{proposition}

Proposition~\ref{proposition:pseudo-inverse} carries two consequences. 
First, coarsening can be done in linear time. Each coarsening level (both in the forward and backward directions) entails multiplication by a sparse matrix. Furthermore, both $P_{\ell}$ and $P_{\ell}^+$ have only $N_{\ell-1}$ non-zero entries meaning that $O(N)$ and $O(M)$ operations suffice to coarsen respectively a vector and a matrix $L$ whose sparsity structure reflects the graph adjacency. 
In addition, the number of graph edges also decreases at each level. Denoting by $\mu_\ell$ the average number of edges of the graphs induced by contraction sets $\mathcal{V}^{(r)}_{\ell-1}$ for every $r$, 
then a quick calculation reveals that the coarsest graph has $m = M - \sum_{\ell=1}^c N_\ell \mu_\ell$ edges.   
If, for instance, at each level all nodes are perfectly contracted into pairs then $\mu_{\ell} = 2$ and $N_{\ell} = N / 2^{\ell}$, meaning that $m = M - 2 N (1 -2^{-c})$.

\subsection{Laplacian consistent coarsening}
\label{subsec:consistent_coarsening}

A further restriction that can be imposed is that coarsening is consistent w.r.t. the Laplacian form. Let $L$ be the combinatorial Laplacian of $G$ defined as
\begin{align}
L(i,j) = 
\begin{cases}
\deg_i & \mbox{if}\ i = j \\
-w_{ij} & \mbox{if}\ e_{ij} \in \E \\
0 & \mbox{otherwise},
\end{cases}
\notag 
\end{align}
where $w_{ij}$ is the weight associated with edge $e_{ij}$ and $\deg_i$ the weighted degree of $v_i$. The following lemma can then be proven:
 
\begin{proposition}[Consistency]
    Let $P$ be a coarsening matrix w.r.t. a graph with combinatorial Laplacian $L$. Matrix $\c{L} = P^\mp L P^+$ is a combinatorial Laplacian if and only if the non-zero entries of $P^+$ are equally valued.
    \label{proposition:proper_laplacian}
\end{proposition}

It is a corollary of Propositions~\ref{proposition:pseudo-inverse} and~\ref{proposition:proper_laplacian} that in consistent coarsening, for any $v'_r \in \V_{\ell}$ and $v_i \in \V_{\ell-1}$ matrices $P_{\ell} \in \Rbb^{N_{\ell} \times N_{\ell-1}}$  and $P^+_{\ell} \in \Rbb^{N_{\ell-1} \times N_{\ell}}$ should be given by:
\begin{align}
    P_{\ell}(r,i) = 
    \begin{cases} 
   \frac{1}{|\V_{\ell-1}^{(r)}|} & \text{if } v_i \in \V_{\ell-1}^{(r)} \\
   0       & \text{otherwise} 
  \end{cases}
  \quad \text{and} \quad
[P_{\ell}^+](i,r) = 
    \begin{cases} 
   1 \hspace{6mm} & \text{if } v_i \in \V_{\ell-1}^{(r)} \\
   0       & \text{otherwise}, 
  \end{cases}
  \notag
\end{align}
where the contraction sets $\V_{\ell-1}^{(1)}, \ldots, \V_{\ell-1}^{(N_{\ell})}$ were defined in Section~\ref{subsec:graph_coarsening}. 

\begin{figure}[t]
	\centering
	\subfloat[Graph $G$]{
		\includegraphics[width=0.4\columnwidth,trim={0.0cm 0.38cm 0.0cm 0.5cm},clip]{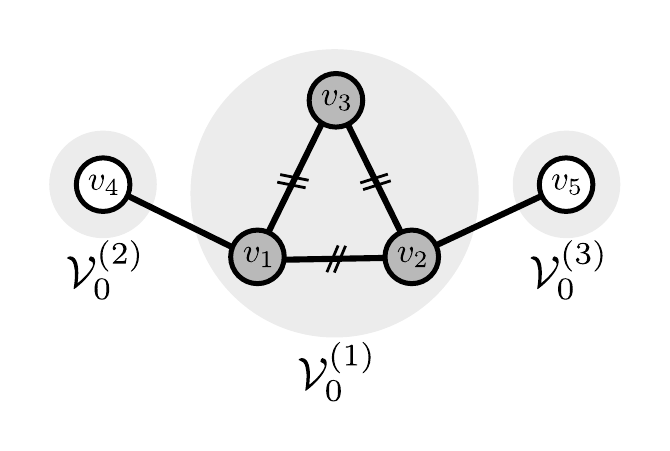}\label{fig:example1}
	}
	\hfill
	\subfloat[Coarse graph $\c{G}$]{
		\includegraphics[width=0.4\columnwidth,trim={0.0cm 0.10cm 0.0cm 0.6cm},clip]{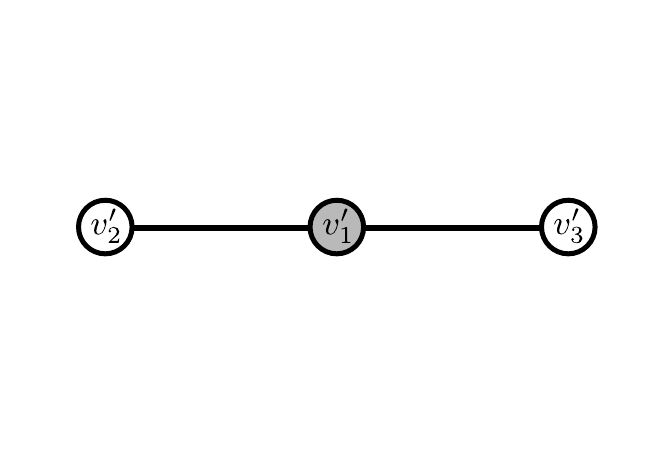}\label{fig:example2}
	}
	\vspace{2mm}
	\caption{Toy coarsening example. Grey discs denote contraction sets. The first three vertices of $G$ forming contraction set $\mathcal{V}_0^{1}$ are contracted onto vertex $v_1'$. All other vertices remain unaffected. \label{fig:example}}
\end{figure}

The toy graph shown in Figure~\ref{fig:example1} illustrates an example where the gray vertices $\mathcal{V}_0^{(1)} = \{ v_1, v_2, v_3\}$ of $G$ are coarsened into vertex $v_1'$, as shown in Figure~\ref{fig:example2}. The main matrices I have defined are
\begin{align*}
P_1 = 
\begin{bmatrix} 
\sfrac{1}{3} & \sfrac{1}{3} & \sfrac{1}{3} & 0 & 0 \\ 
0 & 0 & 0  & 1 & 0 \\
0 & 0 & 0 & 0 & 1 
\end{bmatrix}
\ 
\ 
P_1^+ = 
\begin{bmatrix} 
1 & 0 & 0  \\ 
1 & 0 & 0  \\
1 & 0 & 0  \\ 
0 & 1 & 0  \\ 
0 & 0 & 1  
\end{bmatrix}
\ \ 
\Pi = P_1^+ P_1 = 
\begin{bmatrix} 
\sfrac{1}{3} & \sfrac{1}{3} & \sfrac{1}{3} & 0 & 0  \\ 
\sfrac{1}{3} & \sfrac{1}{3} & \sfrac{1}{3} & 0 & 0  \\
\sfrac{1}{3} & \sfrac{1}{3} & \sfrac{1}{3} & 0 & 0  \\ 
0 & 0 & 0 & 1 & 0  \\ 
0 & 0 & 0 & 0 & 1  
\end{bmatrix}
\end{align*}
and coarsening results in
\begin{align}
\c{L} = P_1^\mp L P_1^+ =
\begin{bmatrix} 
2 & -1 & -1 \\ 
-1  & 1 & 0 \\
-1  & 0 & 1 
\end{bmatrix} 
\quad 
\c{x} = P_1 x = 
\begin{bmatrix} 
(x(1) + x(2) + x(3))/3 \\ 
x(4) \\
x(5)  
\end{bmatrix}.
\notag 
\end{align}
Finally, when lifted $\c{x}$ becomes
\begin{align*}
\p{x} = P_1^+ \c{x} = 
\begin{bmatrix} 
\sfrac{\left(x(1) + x(2) + x(3)\right)}{3} \\ 
\sfrac{\left(x(1) + x(2) + x(3)\right)}{3} \\ 
\sfrac{\left(x(1) + x(2) + x(3)\right)}{3} \\ 
x(4) \\
x(5)  
\end{bmatrix}.
\notag 
\end{align*}

Since vertices $v_4$ and $v_5$ are not affected, the respective contraction sets $\mathcal{V}_0^{(2)}$ and $\mathcal{V}_0^{(3)}$ are singleton sets.

\subsection{Properties of Laplacian consistent coarsening} 

Due to its particular construction, Laplacian consistent coarsening is accompanied by a number of interesting properties. We lay out three in the following:

\emph{Cuts.} To begin with, weights of edges in $G_c$ correspond to weights of cuts in $G$. 

\begin{property}
For any level $\ell$, the weight $W_\ell(r,q)$ between vertices $v'_r , v'_q \in \V_{\ell}$ is equal to 
$$W_\ell(r,q) = \sum_{v_i \in \mathcal{S}_\ell^{(r)}} \sum_{v_j \in \mathcal{S}_\ell^{(q)} } w_{ij},$$
where $\S_\ell^{(r)} = \{ v_i \in \V : \varphi_\ell \circ \cdots \circ \varphi_1(v_i) = v'_r \} \subset \V$ contains all vertices of $G$ contracted onto $v'_r \in \V_\ell
$. 
\label{claim:cuts}
\end{property}

In the toy example, there exists a single edge of unit weight connecting vertices in $\mathcal{V}_0^{(1)}$ and $\mathcal{V}_0^{(2)}$, and as such the weight between $v_1'$ and $v_2'$ is equal to one.

\emph{Eigenvalue interlacing.} For a single level of Laplacian consistent coarsening, matrix $P P^\top = P_1 P_1^\top$ is given by $\diag{1/|\V_{0}^{(1)}|, \ldots, 1/|\V_{0}^{(N_1)}|}$, implying that the multiplicative constants in Theorem~\ref{theorem:interlacing} are:
$$\gamma_1 = \min_{v_i \in \V} |\V_{0}^{ \varphi_1(v_i)}| \geq 1 \quad \text{and} \quad \gamma_2 = \max_{v_i \in \V} |\V_{0}^{\varphi_1(v_i)}|.$$ 
Above, $v'_r = \varphi_1(v_i) \in \V_1$ is the vertex to which $v_i$ is mapped to and the set $\V_{0}^{\varphi_1(v_i)}$ contains all vertices also contracted to $v'_r$. Thus in the toy example, $\lambda_k \leq \tilde{\lambda}_k \leq 3 \lambda_{k + 2}$ for every $k \leq 3$.
If multiple levels are utilized these terms become dependent on the sequence of contractions. 
To obtain a general bound let $\varphi_{1}^{\ell}(v_i)= \varphi_{\ell} \circ \cdots \circ \varphi_{1} (v_i) \in \V_{\ell}$ be the vertex onto which $v_i \in \V$ is contracted to in the $\ell$-th level. 
\begin{property}
If $L_c$ is obtained from $L$ by Laplacian consistent coarsening, then
$$ \gamma_1 \geq \min_{v_i \in \V } \, \prod_{\ell=1}^c |\V_{\ell-1}^{\varphi_{1}^{\ell}(v_i)}| \geq 1 \quad \text{and} \quad  \gamma_2 \leq \max_{v_i \in \V } \, \prod_{\ell=1}^c |\V_{\ell-1}^{\varphi_{1}^{\ell}(v_i)}|, $$
with the set $\V_{\ell-1}^{\varphi_{1}^{\ell}(v_i)} $ containing all vertices of $\V_{\ell-1}$ that are contracted onto $\varphi_{1}^{\ell}(v_i)$.
\end{property}
Though not included the proof follows from the diagonal form of $P_{\ell} \cdots P_{1} P_1^\top \cdots P_{\ell}^\top$ and the special row structure of each $P_{\ell}$ for every $\ell$. The dependency of $\tilde{\lambda}_k$ on the size of contraction sets can be removed either by enforcing at each level that all contraction sets have identical size and dividing the graph weights by that size, or by re-normalizing each $P_{\ell}$ such that $P_{\ell}^\top = P_{\ell}^+$. The latter approach was used by~\cite{loukas2018spectrally} but is not adopted here as it results in $L_c$ losing its Laplacian form.

\emph{Nullspace.} Finally, as is desirable, the structure of the nullspace of $L$ is preserved both by coarsening and lifting:
\begin{property}
If $P$ is a (multi-level) Laplacian consistent coarsening matrix, then
\begin{align}
    P \1_{N} = \1_{n} \quad \text{and} \quad P^+ \1_{n} = \1_{N}, \notag
\end{align}
where the subscript indicates the dimensionality of the constant vector. 
\end{property}
Thus, we can casually ignore vectors parallel to the constant vector in our analysis.

\section{Restricted notions of approximation}

This section aims to formalize how should a graph be reduced such that the structure of the reduced and original problems should be as close as possible. Inspired by work in graph sparsification, I introduce a measure of approximation that is tailored to graph reduction. The new definition implies strong guarantees about the distance of the original and coarsened spectrum and gives conditions such that the cut structure of a graph is preserved by coarsening.

\subsection{Restricted spectral approximation}
\label{subsec:subspace_similarity}

One way to define how close a PSD matrix $L$ is to its reduced counterpart 
is to establish an isometry guarantee w.r.t. the following induced semi-norms:
\begin{align}
    \norm{x}_L = \sqrt{x^\top L x} \quad \text{and} \quad \norm{\c{x}}_{\c{L}} = \sqrt{\c{x}^\top \c{L} \c{x}} \notag
\end{align}

Ideally, one would hope that there exists $\epsilon > 0$ such that
\begin{align}
    (1 - \epsilon) \, \norm{x}_L \leq \norm{\c{x}}_{\c{L}} \leq (1 + \epsilon) \, \norm{x}_L
\label{eq:spectral_similarity2}
\end{align}
for all $ x \in \Rbb^N$. 

If the equation holds, matrices $\c{L}$ and $L$ are called $\epsilon$-similar.
The objective of constructing sparse spectrally similar graphs is the main idea of spectral graph sparsifiers, a popular method for accelerating the solution of linear systems involving the Laplacian. In addition, spectral similarity carries a number of interesting consequences that are of great help in the construction of approximation algorithms: the eigenvalues and eigenvectors of two similar graphs are close and, moreover, all vertex partitions have similar cut size.  

In contrast to graph sparsification however, since here the dimension of the space changes it is impossible to satisfy~\eqref{eq:spectral_similarity2} for {every} $x \in \Rbb^N$ unless one trivially sets $\epsilon = 1$ (this follows by a simple rank argument). 
To carry out a meaningful analysis, one needs to consider a subspace of dimension $k \leq n$ and aim to approximate the behavior of $L$ solely within it. 

I define the following generalization of spectral similarity:

\begin{definition}[Restricted spectral approximation]
	Let $\mathbf{R}$ be a $k$-dimensional subspace of $\mathbb{R}^N$. 
Matrices $L_c$ and $L$ are $(\mathbf{R}, \epsilon)$-similar if there exists an $\epsilon\geq 0$ such that
	\begin{align}
	\norm{x - \tilde{x}}_{L} \leq \epsilon \norm{x}_L, \notag \quad \text{for all} \quad x \in \mathbf{R}, 
	\end{align}
	where $\tilde{x}= P^+ P x$.
\end{definition}
In addition to the restriction on $\mathbf{R}$, the above definition differs from~\eqref{eq:spectral_similarity2} in the way error is measured. In fact, it asserts a property that is slightly stronger than an approximate isometry w.r.t. a semi-norm within $\mathbf{R}$. The strengthening of the notion of approximation deviates from the restricted spectral similarity property proposed by~\cite{loukas2018spectrally} and is a key ingredient in obtaining multi-level bounds. Nevertheless, one may recover a restricted spectral similarity-type guarantee as a direct consequence: 
\begin{corollary}
	If $L_c$ and $L$ are $(\mathbf{R},\epsilon)$-similar, then  
	$$
	(1 - \epsilon) \, \norm{x}_L \leq \norm{\c{x}}_{\c{L}} \leq (1 + \epsilon) \, \norm{x}_L,  \notag \quad \text{for all} \quad x \in \mathbf{R}.
	$$
	\label{corollary:restricted_isometry}
\end{corollary}
\begin{proof}
Let $S$ be defined such that $L = S^\top  S$. By the triangle inequality:
\begin{align}
	|\norm{x}_L - \norm{\c{x}}_{\c{L}}| = |\norm{Sx} - \norm{S P^+ P x}_2| \leq \| Sx - S P^+ P x \|_2 = \| x - \tilde{x} \|_L \leq \epsilon \norm{x}_L, \notag 
\end{align}
which is equivalent to the claimed relation.
\end{proof}
%


Clearly, if $L_c$ and $L$ are $(\mathbf{R}, \epsilon)$-similar then they are also $({\mathbf{R}}', \epsilon')$-similar, where ${\mathbf{R}}'$ is any subspace of $\mathbf{R}$ and $\epsilon'\geq \epsilon$. As such, results about large subspaces and small $\epsilon$ are the most desirable.

It will be shown in Sections~\ref{subsec:spectrum} and~\ref{subsec:cuts} that the above definition implies restricted versions of the spectral and cut guarantees provided by spectral similarity. For instance, instead of attempting to approximate the entire spectrum as done by spectral graph sparsifiers, here one can focus on a subset of the spectrum with particular significance.

\subsection{Implications for the graph spectrum}
\label{subsec:spectrum}

One of the key benefits of restricted spectral approximation is that it implies a relation between the spectra of matrices $L$ and $\c{L}$ that goes beyond interlacing (see Theorem~\ref{theorem:interlacing}). 

To this effect, consider the smallest $k$ eigenvalues and respective eigenvectors and define the following matrices:
\begin{align}
     U_k \in \Rbb^{N\times k} = [u_1, u_2, \ldots, u_k] \quad \text{and} \quad \Lambda_k = \diag{\lambda_1, \lambda_2, \ldots, \lambda_k} \notag
\end{align}
As I will show next, ensuring that $\epsilon$ in Proposition~\ref{proposition:restricted_similarity} is small when $\mathbf{R} = \mathbf{U}_k \defequal \spanning{U_k}$ suffices to guarantee that the first $k$ eigenvalues and eigenvectors of $L$ and $L_c$ are aligned.

The first result concerns eigenvalues.

\begin{theorem}[Eigenvalue approximation]
If $L_c$ and $L$ are $(\mathbf{U}_k,\epsilon_k)$-similar, then 
\begin{align}
\gamma_1 \, \lambda_k \leq \p{\lambda}_k 
&\leq \gamma_2 \, \frac{ \left(1 + \epsilon_k\right)^2}{1 - \epsilon_k^2 (\lambda_k/\lambda_2) } \, \lambda_k \notag,
\end{align}
whenever $\epsilon_k^2 < \lambda_2/\lambda_k$.
\label{theorem:eigenvalues}
\end{theorem}
Crucially, the bound depends on $\lambda_k$ instead of $\lambda_{k+N-n}$ and thus can be significantly tighter than the one given by Theorem~\ref{theorem:interlacing}. Noticing that $\epsilon_k \leq \epsilon_{k'}$ whenever $k < k'$, one also deduces that it is stronger for smaller eigenvalues. For $k=2$ in particular, one has
\begin{align}
\gamma_1 \, \lambda_2 \leq \p{\lambda}_2 
\leq \gamma_2 \, \frac{\left(1 + \epsilon_2\right)^2}{1 - \epsilon_2^2} \, \lambda_2, \notag
\end{align}
%
which is small when $\epsilon_2 \ll 1$.


I also analyze the angle between principal eigenspaces of $L$ and $\c{L}$.  I follow~\cite{li1994relative} and split the eigendecompositions of $L =  U \Lambda U^\top$ and $P^\top \c{L} P = P^\top \p{U} \p{\Lambda} \p{U}^\top P$ as
\begin{align*}
L &= (U_k, U_{k^\bot})
	\begin{pmatrix}
		\Lambda_k &  \\ 
		 & \Lambda_{k^\bot} 
	\end{pmatrix}
	\begin{pmatrix}
		U_k^\top  \\ 
		U_{k^\bot}^\top 
	\end{pmatrix} 
	\quad 
	 P^\top \c{L} P = (P^\top \p{U}_k, P^\top\p{U}_{k^\bot})
	 \begin{pmatrix}
	 \p{\Lambda}_k &  \\ 
	 & \p{\Lambda}_{k^\bot} 
	 \end{pmatrix}
	 \begin{pmatrix}
	 \p{U}_k^\top P \\ 
	 \p{U}_{k^\bot}^\top P
	 \end{pmatrix},
\end{align*}
where $\p{\Lambda}_k$ and $\p{U}_k$ are defined analogously to $\Lambda_k$ and $ U_k$.
\citet{davis1970rotation} defined the \emph{canonical angles} between the spaces spanned by $U_k$ and $P^\top \tilde{U}_k$ as the singlular values of the matrix
\begin{align}
	\Theta( U_k, P^\top\p{U}_k) \defequal 
	\arccos(U_k^\top P^\top \p{U}_k \p{U}_k^\top P U_k)^{-\sfrac{1}{2}}, \notag
\end{align}
see also~\citep{stewart1990matrix}.
The smaller the sinus of the canonical angles are the closer the two subspaces lie. 
The following theorem reveals a connection between the Frobenius norm of the sinus of the canonical angles and restricted spectral approximation. 

\begin{theorem}[Eigenspace approximation]
If $L_c$ and $L$ are $(\mathbf{U}_k,\epsilon_k)$-similar then 
\begin{align}
	\norm{ \sintheta{U_k}{P^\top \p{U}_k} }_F^2 
	&\leq \frac{ 1 }{\lambda_{k+1} - \lambda_{k}} \left( \sum\limits_{i \leq k} \lambda_i \left(\frac{(1 + \epsilon_i)^2}{\gamma_1}  - 1\right)  + \lambda_k \sum_{i \leq k} \epsilon_i \right) , \notag  
\end{align}
\label{theorem:sintheta}
\end{theorem}
Note that the theorem above utilizes all $\epsilon_i$ with $i \leq k $, corresponding to the restricted spectral approximation constants for $\mathbf{R} = \mathbf{U}_i$, respectively. However, all these can be trivially relaxed to $\epsilon_k$, since $\epsilon_i\leq \epsilon_k$ for all $i \leq k$.
%

\subsection{Implications for graph partitioning}
\label{subsec:cuts}

One of the most popular applications of coarsening is to accelerate graph partitioning~\citep{hendrickson1995multi,karypis1998fast,kushnir2006fast,dhillon2007weighted,wang2014partition}. In the following, I provide a rigorous justification for this choice by showing that if the (Laplacian consistent) coarsening is done well and $G_c$ contains a good cut, then so will $G$. For the specific case of spectral clustering, I also provide an explicit bound on the coarse solution quality.

\paragraph{Existence results.} For consistent coarsening, the spectrum approximation results presented previously imply similarities between the cut-structures of $G_c$ and $G$. 

To formalize this intuition, the
conductance of any subset $\mathcal{S}$ of $\mathcal{V}$ is defined as
\begin{align}
     \phi(\mathcal{S}) \defequal \frac{ w(\mathcal{S}, \bar{\mathcal{S}}) }{ \min\{w(\mathcal{S}), w(\bar{\mathcal{S}})\} }, \notag 
\end{align}
where $\bar{\mathcal{S}} = \mathcal{V} \setminus \mathcal{S}$ is the complement set, $w(\mathcal{S}, \bar{\mathcal{S}}) = \sum_{v_i\in \mathcal{S}, v_j \in \bar{\mathcal{S}}} w_{ij}$ is the weight of the cut and $w(\mathcal{S}) = \sum_{v_i \in \S} \sum_{v_j \in \V} w_{ij} $ is the volume of $\mathcal{S}$. 

The \emph{$k$-conductance} of a graph measures how easy it is to cut it into $k$ disjoint subsets $\mathcal{S}_1, \ldots, \mathcal{S}_k \subset \mathcal{V}$ of balanced volume: 
\begin{align}
	\phi_k(G) = \min_{\mathcal{S}_1, \ldots, \mathcal{S}_k} \max_{i} \phi(\mathcal{S}_i) \notag 
\end{align}
The smaller $\phi_k(G)$ is, the better the partitioning. 

As it turns out, restricted spectral approximation can be used to relate the conductance of the original and coarse graphs. To state the result, it will be useful to denote by $D$ the diagonal degree matrix and further to suppose that $V_{k}$ contains the first ${k}$ eigenvectors of the normalized Laplacian $L^n = D^{-\sfrac{1}{2}} L D^{-\sfrac{1}{2}}$, whose eigenvalues are $0 = \mu_1 
\leq \cdots \leq \mu_k$.

\begin{theorem} 
For any graph $G$ and integer $2\leq k\leq \left\lfloor n/2 \right\rfloor$, if $L_c$ and $L$ are $(\mathbf{R}_{2k}, \epsilon_{2k})$-similar combinatorial Laplacian matrices then
$$\phi_{k}(G) \leq \phi_{k}(G_c) 
	= O\left( \sqrt{ \frac{ \gamma_2 \, (1 + \epsilon_{2k})^2 \xi_{k}(G) }{1 - \epsilon_{2k}^2 (\mu_{2k}/\mu_2)} \phi_k(G)} \right)$$
with $\mathbf{R}_{2k} = \spanning{D^{-\sfrac{1}{2}} V_{2k}}$ and $\xi_k(G) = \log k$, whenever $\epsilon_{2k}^2 < \mu_2/\mu_{2k}$. If $G$ is planar then $\xi_{k}(G) = 1.$ More generally, if $G$ excludes $K_h$ as a minor, then $\xi_{k}(G) = h^4.$ For $k=2$, supposing that $L_c$ and $L$ are $(\mathbf{R}_2, \epsilon_{2})$-similar, we additionally have
$$
\phi_{2}(G) \leq \phi_{2}(G_c) \leq 2 \sqrt{ \frac{ \gamma_2 \, (1 + \epsilon_{2})^2 }{1 - \epsilon_{2}^2} \phi_2(G) }.
$$
\label{theorem:cheeger}
\end{theorem}

This is a non-constructive result: it does not reveal how to find the optimal partitioning, but provides conditions such that the latter is of similar quality in the two graphs.

\paragraph{Spectral clustering.}
It is also possible to derive approximation results about the solution quality of unsupervised learning algorithms that utilize the first $k$ eigenvectors in order to partition $G$. 
I focus here on spectral clustering. To perform the analysis, let $U_k$ and $P^\top \p{U}_k$ be the spectral embedding of the vertices w.r.t. $L$ and $L_c$, respectively, and define the optimal partitioning as 
\begin{align} 
\mathcal{P}^{*} = \argmin_{\mathcal{P} = \{\mathcal{S}_1, \ldots, \mathcal{S}_k\}} \kmeans{k}{U_k}{\mathcal{P}} \ \  \text{and} \ \ \p{\mathcal{P}}^{*} = \argmin_{\mathcal{P} = \{\mathcal{S}_1, \ldots, \mathcal{S}_k\}} \kmeans{k}{P^\top \p{U}_k }{\mathcal{P}},
\label{eq:kmeans}
\end{align}
where, for any embedding $X$, the $k$-means cost induced by partitioning $\V$ into clusters $\S_1, \ldots, \S_k$ is defined as $$ \kmeans{k}{X}{\mathcal{P}} \defequal \sum_{z = 1}^k \sum_{v_i, v_j \in \mathcal{S}_z } \frac{\| X(i,:) - X(j,:)\|_2^2}{2\, |\mathcal{S}_z|}. $$
One then measures the quality of $\p{\mathcal{P}}^{*}$ by examining how far the correct minimizer $\kmeans{k}{U_k}{\mathcal{P}^{*}}$ is to $\kmeans{k}{U_k}{\p{\mathcal{P}}^{*}}$. 
\citet{boutsidis2015spectral} noted that if the two quantities are close then, despite the clusters themselves possibly being different, they both feature the same quality with respect to the $k$-means objective. 

An end-to-end control of the $k$-means error is obtained by combining the inequality derived by~\citet{loukas2018spectrally}, based on the works of~\citep{boutsidis2015spectral,yu2014useful,martin2017fast},
$
| \kmeans{k}{U_k}{\mathcal{P}^{*}}^{\sfrac{1}{2}} - \kmeans{k}{U_k}{\p{\mathcal{P}}^{*}}^{\sfrac{1}{2}} | \leq 2 \sqrt{2} \, \norm{\sintheta{U_k}{P^\top \p{U}_k} }_F
$
with Theorem~\ref{theorem:sintheta}: 

\begin{corollary}
If $L_c$ and $L$ are $(\mathbf{U}_k,\epsilon_k)$-similar then 
\begin{align}	
	\left( \kmeans{k}{U_k}{\mathcal{P}^{*}}^{\sfrac{1}{2}} - \kmeans{k}{U_k}{\p{\mathcal{P}}^{*}}^{\sfrac{1}{2}} \right)^2 \leq \frac{ 8 }{\lambda_{k+1} - \lambda_{k}} \left( \sum\limits_{i \leq k} \lambda_i \left(\frac{(1 + \epsilon_i)^2}{\gamma_1}  - 1\right)  + \lambda_k \sum_{i \leq k} \epsilon_i \right). \notag  
\end{align}
\label{corollary:spectral_clustering}
\end{corollary}
Contrary to previous analysis~\citep{loukas2018spectrally}, the approximation result here is applicable to any number of levels and it can be adapted to hold for the eigenvectors of the normalized Laplacian\footnote{For the normalized Laplacian, one should perform (combinatorial) Laplacian consistent coarsening on a modified eigenspace, as in the proof of Theorem~\ref{theorem:cheeger}.}. 
Nevertheless, it should be stressed that at this point it is an open question whether the above analysis yields benefits over other approaches tailored especially to the acceleration of spectral clustering. A plethora of such specialized algorithms are known~\citep{tremblay2016compressive,boutsidis2015spectral}---arguing about the pros and cons of each extends beyond the scope of this work.  

One might be tempted to change the construction so as to increase $\gamma_1$. For example, this could be achieved by multiplying $P$ with a small constant (see Theorem~\ref{theorem:interlacing}). In reality however, such a modification would not yield any improvement as the increase of $\gamma_1$ would also be accompanied by an increase of $\epsilon_i$.

\subsection{Some limits of restricted spectral approximation}
 
The connection between spectral clustering and coarsening runs deeper than what was shown so far. As it turns out, the first $k$ restricted spectral approximation constants $\epsilon_1, \ldots, \epsilon_k$ associated with a Laplacian consistent coarsening are linked to the $n$-means cost $\kmeans{n}{U_k}{\mathcal{P}}$ induced by the contraction sets $\mathcal{P} = \{\V_0^{(1)}, \ldots, \V_0^{(n)}\}$. The following lower bound is a direct consequence: 
\begin{proposition} 
Let $L$ be a Laplacian matrix. For any $L_c$ obtained by a single level of Laplacian consistent coarsening, if $L_c$ and $L$ are $(\mathbf{U}_k, \epsilon_k)$-similar then it must be that
\begin{align}
	\sum_{i \leq k} \epsilon_k \geq \kmeans{n}{U_k}{\mathcal{P}^*}, \notag
\end{align}
with $\kmeans{n}{U_k}{\mathcal{P}^*}$ being the optimal $n$-means cost for the points $U_k(1,:), \ldots, U_k(N,:)$.
\label{proposition:lower_bound}
\end{proposition}
Computing the aforementioned lower bound is known to be NP-hard, so the result is mostly of theoretical interest.

\section{Graph coarsening by local variation}
\label{sec:algorithms}

This section proposes algorithms for Laplacian consistent graph coarsening. I suppose that $L$ is a combinatorial graph Laplacian 
and, given subspace $\mathbf{R}$ and target graph size $n$, aim to find an $(\mathbf{R}, \epsilon)$-similar Laplacian $L_c$ of size $ n \times n$ with $\epsilon$ smaller than some threshold $\epsilon'$.

Local variation algorithms differ only in the type of contraction sets that they consider. For instance,  the edge-based local variation algorithm only contracts edges, whereas in the neighborhood-based variant each contraction set is a subsets of the neighborhood of a vertex. Otherwise, all local variation algorithms follow the same general methodology and aim to minimize an upper bound of $\epsilon$.
To this end, two bounds are exploited: First, $\c{L}$ is shown to be $(\mathbf{R}, \epsilon)$-similar to $L$ with $\epsilon \leq \prod_{\ell} (1 + \sigma_\ell) - 1$, where the \emph{variation cost} $\sigma_{\ell}$ depends only on previous levels (see Section~\ref{subsec:sufficient}).
The main difficulty with minimizing $\sigma_\ell$ is that it depends on interactions between contraction sets. For this reason, the second bound shows that these interactions can be decoupled by considering each \emph{local variation cost}, i.e., the cost of contracting solely the vertices in $\V^{(r)}_{\ell-1}$, independently on a slightly modified subgraph (see Section~\ref{subsec:decoupling}). Having achieved this, Section~\ref{subsec:families} considers ways of efficiently identifying disjoint contraction sets with small local variation cost.

\subsection{Decoupling levels and the variation cost} 
\label{subsec:sufficient}

Guaranteeing restricted spectral approximation w.r.t. subspace $\mathbf{R}$ boils down to minimizing at each level $\ell$ the \emph{variation cost} $$\sigma_{\ell} = \|\Pi_{\ell}^\bot A_{\ell-1}\|_{L_{\ell-1}} = \|S_{\ell-1} \Pi_{\ell}^\bot A_{\ell-1}\|_2,$$ where $L_{\ell-1} = S_{\ell-1}^\top S_{\ell-1}$ and $ \Pi_{\ell}^\bot = I - P_{\ell}^+ P_\ell$ is a projection matrix. 
Matrix $A_{\ell-1}$ captures two types of information: 
\vspace{-1mm}
\begin{enumerate}
\item Foremost, it encodes the behavior of the target matrix $L$ w.r.t. $\mathbf{R}$. This is clearly seen in the first level, for which one has that $A_{0} = V V^\top L^{+\sfrac{1}{2}}$ with $V \in \Rbb^{N\times k}$ being an orthonormal basis of $\mathbf{R}$.
\item When $\ell>1$ one needs to consider $A_{0}$ in view of the reduction done in previous levels. The necessary modification turns out to be $A_{\ell-1} = B_{\ell-1} (B_{\ell-1}^\top L_{\ell-1} B_{\ell-1})^{+\sfrac{1}{2}}$, with $B_{\ell-1} = P_{\ell-1} B_{\ell-2} \in \Rbb^{N_{\ell-1} \times N}$ expressed in a recursive manner and $B_{0} = A_0$.
\end{enumerate}

The following result makes explicit the connection between $\epsilon$ and $\sigma_{\ell}$.
\begin{proposition}
Matrices $L_c$ and $L$ are $(\mathbf{R}, \epsilon)$-similar with 
$
\epsilon \leq \prod_{\ell=1}^{c} ( 1 + \sigma_\ell) -1.
$
\label{proposition:restricted_similarity}
\end{proposition}

Crucially, the above makes it possible to design a multi-level coarsening greedily, by starting from the first level and optimizing consecutive levels one at a time: 
%
 
\begin{algorithm}[h!]
\caption{\textsf{Multi-level coarsening}}
\label{algorithm:multi-level}
\begin{algorithmic}[1]
\State \textbf{input}: Combinatorial Laplacian $L$, threshold $\epsilon'$, and target size $n$.
\State Set $\ell \gets 0$, $L_\ell \gets L$, and $\epsilon_\ell \gets 0$.
\While{$ N_{\ell} > n $ and $\epsilon_\ell < \epsilon'$} 
	\State $\ell \gets \ell+1$
	\State Coarsen $L_{\ell-1}$ using Algorithm~\ref{algorithm} with threshold $\sigma'=\frac{1 + \epsilon'}{1 + \epsilon_{\ell-1}}-1$ and target size $n$. Let $L_{\ell}$ be the resulting Laplacian of size $N_{\ell}$ with variation cost $\sigma_{\ell}$.
	\State $\epsilon_{\ell} \gets (1 + \epsilon_{\ell-1})(1 + \sigma_\ell) - 1$.
\EndWhile
\State \textbf{return} $L_{\ell}$
\end{algorithmic}
\end{algorithm}

It is a consequence of Proposition~\ref{proposition:restricted_similarity} that the above algorithm returns a Laplacian matrix $L_c$ that is $(\mathbf{R}, \epsilon)$-similar to $L$ with $\epsilon \leq \epsilon_c \leq \epsilon'$, where $c$ is the last level $\ell$. On the other hand, setting $\epsilon'$ to a large value ensures that the same algorithm always attains the target reduction at the expense of loose restricted approximation guarantees.

\emph{Remark.} The variation cost simplifies when $\mathbf{R}$ is an eigenspace of $L$. I demonstrate this for the choice of $\mathbf{U}_k$, though an identical argument can be easily derived for any eigenspace. Denote by $\Lambda$ the diagonal $N\times N$ eigenvalue matrix placed from top-left to bottom-right in non-decreasing order and by $U$ the respective full eigenvector matrix. Furthermore, let $\Lambda_k$ be the $k\times k$ sub-matrix of $\Lambda$ with the smallest $k$ eigenvalues in its diagonal. By the unitary invariance of the spectral norm, it follows that $\sigma_0 = \|\Pi_{1}^\bot U_k U_k^\top L^{+\sfrac{1}{2}} \|_{L_{0}} = \|\Pi_{1}^\bot U_k U_k^\top L^{+\sfrac{1}{2}} U \|_{L_0} = \|\Pi_{1}^\bot U_k U_k^\top U \Lambda^{+\sfrac{1}{2}} \|_{L_0}$ . Simplifying and eliminating zero columns, one may redefine 
$B_{0} = U_k \Lambda_k^{+\sfrac{1}{2}} \in \Rbb^{N \times k}$, such that once more $\sigma_0 = \|\Pi_1^\bot B_0\|_{L_0}$. 
This is computationally attractive because now at each level one needs to take the pseudo-inverse-square-root of a $k\times k$ matrix $B_{\ell-1}^\top L_{\ell-1} B_{\ell-1}$, with $k \ll N$.


\subsection{Decoupling contraction sets and local variation}
\label{subsec:decoupling}

Suppose that $\Pi_{\mathcal{C}}^\bot$ is the (complement) projection matrix obtained by contracting solely the vertices in set $\mathcal{C}$, while leaving all other vertices in $\V_{\ell-1}$ untouched:
$$
\left[ \Pi^\bot_{\mathcal{C}} \, x \right](i) = 
\begin{cases}
x(i) - \sum_{v_j \in \mathcal{C}} \frac{x(j)}{|\mathcal{C}|} & \mbox{if}\ v_i \in \mathcal{C} \\
0 & \mbox{otherwise}.
\end{cases}
$$
(Here, for convenience, the level index is suppressed.)

Furthermore, let $L_\mathcal{C}$ be the $N_{\ell-1} \times N_{\ell-1}$ combinatorial Laplacian whose weight matrix is
\begin{align}
\left[W_\mathcal{C}\right](i,j) = 
\begin{cases}
W_{\ell-1}(i,j) & \mbox{if}\ v_i, v_j \in \mathcal{C} \\
2\,W_{\ell-1}(i,j) & \mbox{if}\ v_i \in \mathcal{C} \ \text{and} \ v_j \notin \mathcal{C} \\
0 & \mbox{otherwise}.
\end{cases}
\end{align}
That is, $W_\mathcal{C}$ is zero everywhere other than at the edges touching at least one vertex in $\mathcal{C}$.  
The following proposition shows us how to decouple the contribution of each contraction set to the variation cost.

\begin{proposition}
The variation cost is bounded by
$$\sigma_\ell^2 \leq \sum_{ \mathcal{C} \in \mathcal{P}_\ell } \| \Pi_\mathcal{C}^\bot \, A_{\ell-1}\|_{L_\mathcal{C} }^2 ,$$
where $\mathcal{P}_{\ell} = \{ \V_{\ell-1}^{(1)}, \ldots, \V_{\ell-1}^{(N_{\ell})}\}$ is the family of contraction sets of level $\ell$.
\label{proposition:decoupling}
\end{proposition}

The above argument therefore entails bounding the, difficult to optimize, variation cost as a function of locally computable and independent costs $ \| \Pi_\mathcal{C}^\bot \, A_{\ell-1}\|_{L_\mathcal{C} }^2$. The obtained expression is a relaxation, as it assumes that the interaction between contraction sets will be the worst possible. It might be interesting to notice that the quality of the relaxation depends on the weight of the cut between contraction sets. Taking the limit, the inequality converges to an equality as the weight of the cut shrinks. 
Also of note, the bound becomes tighter the larger the dimensionality reduction requested (the smaller $N_{\ell}= |\mathcal{P}_{\ell}|$ is, the fewer inequalities are involved in the derivation).    

\subsection{Local variation coarsening algorithms}
\label{subsec:families}

Starting from a \emph{candidate family} $\F_{\ell} = \{ \C_{1}, \C_{2}, \C_{3}, \ldots \}$, that is, an appropriately sized family of candidate contraction sets, the strategy will be to search for a small \emph{contraction family} $\mathcal{P}_{\ell} = \{ \V_{\ell-1}^{(1)}, \ldots, \V_{\ell-1}^{(N_{\ell})}\}$ with minimal variation cost $\sigma_\ell$ ($\mathcal{P}_{\ell}$ is valid if it partitions $\V_{\ell-1}$ into $N_{\ell}$ {contraction sets}). 
Every coarse vertex $v_r' \in \V_\ell$ is then formed by contracting the vertices in $\V_{\ell-1}^{(r)}$.

As a thought experiment, suppose that set $\C \in \F_{\ell}$ is chosen to be part of $\mathcal{P}_\ell$. From the decoupling argument, its contribution to $\sigma_\ell^2$ will be at most $\| \Pi_{\C}^\bot \, A_{\ell-1}\|_{L_{\C}}^2$ independently of how the other candidate sets are chosen. Moreover, the selection will yield a reduction of $N_{\ell-1}$ by $|\C|-1$ vertices. 
Thus, one needs to look for the non-singleton candidate sets $\C$ with cost 
\begin{align}
    \text{cost}_{\ell}(\C) \defequal \frac{\| \Pi_{\C}^\bot \, A_{\ell-1}\|_{L_{\C}}^2}{|\C|-1} 
    \label{eq:local_variation_cost}
\end{align}
that is as small as possible. I refer to~\eqref{eq:local_variation_cost} as \emph{local variation cost} because it captures the maximal variation of all signals from an appropriate subspace (implied by $A_{\ell-1}$) with support on $\C$.  
On the other hand, since any permissible contraction family $\mathcal{P}_{\ell}$ should be a partitioning of $\V_{\ell-1}$, choosing $\C$ precludes us from selecting any $\C'$ with which it intersects.    

Based on this intuition, Algorithm~\ref{algorithm} sequentially examines candidate sets from $\F_{\ell}$, starting from those with minimal cost. To decide whether a candidate set $\C$ will be added to $\P_{\ell}$ the algorithm asserts that all vertices in $\C$ are unmarked---essentially enforcing that all contraction sets are disjoint. Accordingly, as soon as $\C$ is added to $\mathcal{P}_{\ell}$, all vertices that are in $\C$ become marked. Candidate sets with marked vertices are pruned ($\C' \gets \C \setminus \textsf{marked}$) and their cost is updated. The algorithm terminates if $\F_{\ell}$ is, if the target reduction is achieved, or if a given error threshold is exceeded. Even though this remains implicit in the discussion, if at termination $\P_{\ell}$ does not cover every vertex of $\V_{\ell-1}$, then I compliment it with singleton sets, featuring one vertex each (and zero cost).

\begin{algorithm}[h!]
\caption{\textsf{Single-level coarsening by local variation}}
\label{algorithm}
\begin{algorithmic}[1]
\State \textbf{input}: Combinatorial Laplacian $L_{\ell-1}$, threshold $\sigma'$, and target size $n$.
\State Form the family of candidate sets $\F_{\ell} = \{ \C_{1}, \C_{2}, \C_{3}, \ldots \}$ (algorithm-specific step).
\State $N_{\ell} \gets N_{\ell-1}$, $\textsf{marked} \gets \varnothing$, $\sigma_{\ell}^2 \gets 0$.
\State Sort $\mathcal{F_{\ell}}$ in terms of increasing $\text{cost}_{\ell}(\C)$.
\While{$|\F_{\ell}|> 0$ and $N_{\ell} > n$ and $\sigma_{\ell} \leq \sigma'$} 
	\State Pop the candidate set $\C$ of minimal cost $s$ from $\F_{\ell}$.
	\If{all vertices of $\C$ are not \textsf{marked} and $\sigma' \geq \sqrt{ \sigma_{\ell}^2 + (|\C| - 1)s}$}
	\State $\textsf{marked} \gets \textsf{marked}\cup \C, \ \mathcal{P}_{\ell} \gets \mathcal{P}_{\ell} \cup \C$, $N_{\ell} \gets N_{\ell} - |\C| + 1$, $\sigma_\ell^2 \gets \sigma_\ell^2 + (|\C| - 1)s$
	\Else{}
	\State $\C' \gets \C \setminus \textsf{marked}$
	    \If{$|\C'| > 1$}
	    \State Compute $\text{cost}_{\ell}(\C')$ and insert $\C'$ into $\mathcal{F_{\ell}}$ while keeping the latter sorted.\label{algorithm:prunning}
	    \EndIf
	\EndIf
\EndWhile
\State Form the $N_{\ell}\times N_{\ell-1}$ coarsening matrix $P_{\ell}$ based on $\mathcal{P}_{\ell}$.
\State \textbf{return} $L_{\ell} \gets P_{\ell}^\mp L_{\ell-1} P_{\ell}^+$ and $\sigma_\ell$
\end{algorithmic}
\end{algorithm}

Undeniably, Algorithm~\ref{algorithm} is only one of the possible ways to select a partitioning of small variation cost. However, this algorithm stands out from other algorithms I experimented with, as it is very efficient when the subspace of interest is an eigenspace (e.g., $V=U_k$), $k$ is small, and the families $\mathcal{F}_{\ell}$ have been selected appropriately.
Denote by $\Phi = \max_{\ell} \sum_{\C \in \F_{\ell}} |\C|$ the maximum number of vertices in all candidate sets and by $\delta = \max_{\ell,\ \C \in \F_{\ell}} |\C|$ the cardinality of the maximum candidate set---I refer to these measures as \emph{family weight} and \emph{width}, respectively.
Choosing $\mathbf{R}=\mathbf{U}_k$, the computational complexity of Algorithm~\ref{algorithm} is $\tilde{O}(ckM + k^2N + ck^3 + c \, \Phi \left( \min \{ k^2 \delta + k \delta^2,\ k\delta^2 + \delta^3\} + \log{\max_{\ell}|\mathcal{F}_{\ell}|} \right) )$, which up to poly-log factors is linear on the number of edges, vertices, and $\Phi$ (see Appendix~\ref{app:complexity} for details).

If computational complexity is of no concern, one may consider the following two more sophisticated algorithms: 
\textit{The optimal algorithm.} Given a candidate family, the algorithm that optimally minimizes the sum of local variation costs constructs a graph with one vertex for each subset of a candidate set and adds an edge between every two vertices whose respective sets have a non-empty intersection. It then selects $\P_{\ell}$ as the maximum independent set of minimal weight (the weight of each vertex is a local variation cost w.r.t a set). Unfortunately, even if the size of this graph is a polynomial on $N$ this problem cannot be solved efficiently, since the minimum-weight independent set problem is NP-hard. Nevertheless, for the specific case where candidate sets correspond to edges the problem simplifies to a minimum-weight matching problem, which can be computed in $O(N_{\ell-1}^3)$ time exactly, whereas a $(2+\delta)$-approximation can be found much faster~\citep{paz20172+}. 
\textit{The quadratic variant.} A second possibility is to proceed as with Algorithm~\ref{algorithm}, but to prune each $\C' \in \F_{\ell}$ after a set $\C$ is added to $\mathcal{P}_{\ell}$. The numerical experiments indicated that this additional step improves slightly the coarsening quality, but it is not recommended for large graphs as it introduces a quadratic dependency of the complexity on $N$.

\paragraph{Candidate contraction families}. To keep coarsening efficient, I focus on families of linear weight and almost constant width. Two possibilities are considered:

\emph{Edge-based.} Here $\F_{\ell}$ contains one candidate set for each edge of $G_{\ell-1}$. This is a natural choice for coarsening---indeed, most coarsening algorithms in the literature use some form of edge contraction. 
It is straightforward to see that in this case $\Phi = 2 M$ and $\delta = 2$, meaning that the expression of the computational complexity simplifies to $\tilde{O}(c k M + ck^3 + k^2N)$. 
The drawback of contracting edges is that at each level the graph size can only reduced by at most a factor of 2, meaning that a large number of levels is necessary to achieve significant reduction\footnote{In practice, depending on the graph in question, the per-level reduction ratio $r_\ell$ is usually between 0.35 and 0.45.}.      

\emph{Neighborhood-based.} A more attractive choice is to construct one candidate set for the neighborhood of each vertex, including the vertex itself. 
Denoting by $\Delta$ the largest combinatorial degree. Since $\Phi = 2 M$, the complexity here is 
$\tilde{O}(c M (k + \min \{ k^2 \Delta + k \Delta^2,\ k\Delta^2 + \Delta^3\}) + ck^3 + k^2N)$. 
Experiments show that the neighborhood-based construction generally achieves better reduction, while being marginally slower than edge-based families.  

As a final remark, when $G$ is dense the dependency on $M$ can be dropped by sparsifying the graph before using Algorithm~\ref{algorithm}.

\section{Numerical results}

The evaluation was performed on four representative graphs, each exhibiting different structural characteristics:
\begin{itemize}
    \item \textit{Yeast}. Protein-to-protein interaction network in budding yeast, analyzed by~\citet{jeong2001lethality}. The network has $N = 1458$ vertices, $M=1948$ edges, diameter of 19, and degree between 1 and 56.
    \item \textit{Airfoil}. Finite-element graph obtained by airflow simulation~\cite{preis1997party}, consisting of $N=4000$ vertices, $M=11490$ edges, diameter of 65, and degree between 1 and 9.
    \item \textit{Minnesota}. Road network with $N=2642$ vertices, $M=3304$ edges, diameter of 99, and degree between 1 and 5~\citep{gleich2008matlabbgl}.
    \item \textit{Bunny}. Point cloud consisting of $N=2503$ vertices, $M=65490$ edges, diameter of 15, and degree between 13 and 97~\citep{turk1994zippered}. The point cloud has been sub-sampled from its original size.
\end{itemize}

I compare to the following methods for multi-level graph reduction:
\begin{itemize}
    \item \textit{Heavy edge matching.} At each level of the scheme, the contraction family is obtained by computing a maximum-weight matching with the weight of each contraction set (i.e., $(v_i,v_j)$) calculated as $w_{ij} / \max\{\deg_i, \deg_j\}$. In this manner, heavier edges connecting vertices that are well separated from the rest of the graph are contracted first. Heavy edge matching was first introduced in the algebraic multigrid literature and, perhaps due its simplicity, it has been repeatedly used for partitioning~\citep{karypis1998fast,dhillon2007weighted} and drawing~\citep{walshaw2000multilevel,hu2005efficient} graphs, as well as more recently in graph convolutional neural networks~\citep{}.    
    
    \item \textit{Algebraic distance.} This method differs from heavy edge matching in that the weight of each contraction set is calculated as $(\sum_{q=1}^Q (x_q(i) - x_q(j))^2)^{1/2}$, where $x_k$ is an $N$-dimensional test vector computed by successive sweeps of Jacobi relaxation. The complete method is described by~\citet{ron2011relaxation}, see also~\citep{chen2011algebraic}. As recommended by the authors, I performed 20 relaxation sweeps. Further, I set the number of test vectors $Q$ to equal the dimension $k$ of the space I aimed to approximate (a simple rank argument shows that $Q\geq k$ for the test vectors to span the space). 
    
    \item \textit{Affinity.} This is a vertex proximity heuristic in the spirit of the algebraic distance that was proposed by~\citet{livne2012lean} in the context of their work on the lean algebraic multigrid. As per author suggestion, the $Q=k$ test vectors are here computed by a single sweep of a Gauss–Seidel iteration.  
    
    \item \textit{Kron reduction.} At each level, the graph size is reduced by selecting a set of vertices of size $N/2$ (corresponding to the positive entries of the last eigenvector of $L$) and applying Kron reduction. The method, which was proposed by~\cite{shuman2016multiscale}, is not strictly a coarsening method as it completely rewires the vertices of the reduced graph, resulting in significantly denser graphs\footnote{As suggested by the authors, the sparsity of reduced graphs can be controlled by spectral sparsification. The sparsification step was not included in the numerical experiments since it often resulted in increased errors.}. Unfortunately, the rewiring step entails finding the Schur complement of a large Laplacian submatrix and thus generally exhibits $O(N^3)$ complexity, rendering it prohibitive for graphs of more than a few thousand vertices. Despite these drawbacks, the method is quite popular because of its elegant theoretical guarantees~\citep{dorfler2013kron}.  
\end{itemize}

Depending on how the edge matching is constructed, different variants of edge contraction methods can be implemented. At the two extremes of the complexity spectrum one finds the maximum matching of minimum weight at a complexity of $O(N^3)$~\citep{galil1986efficient} or greedily constructs a matching by visiting vertices in a random order and inducing $O(M)$ overhead~\citep{dhillon2007weighted}. 

For consistency, I implemented all edge-based methods by combining Algorithm~\ref{algorithm} with an edge-based family and substituting the local variation cost with the (negative) method-specific edge weight. This generally yields matchings of better quality (heavier weight) than visiting vertices in a predefined order, at the expense of the marginally larger $O(M\log{M})$ complexity necessary for sorting the edge weights. The choice is also motivated by the observation that the computational bottleneck of (sophisticated) edge contraction methods lies in the edge weight computation. For all experiments, I set $\epsilon' = \infty$ aiming for a fixed reduction rather than restricted spectral approximation guarantee.
The code reproducing the experiments can be accessed \href{https://www.dropbox.com/s/l8c1r605bdmmvr6/jmlr_code.zip}{\textcolor{blue}{online}}.
 
\subsection{Restricted spectral approximation}

\newcommand{\mys}{0.45\columnwidth}
\begin{figure*}[!h]
\centering
  \subfloat[yeast ($k=10$)]{\includegraphics[width=\mys,trim={0.8cm 1cm 0.8cm 1cm},clip]{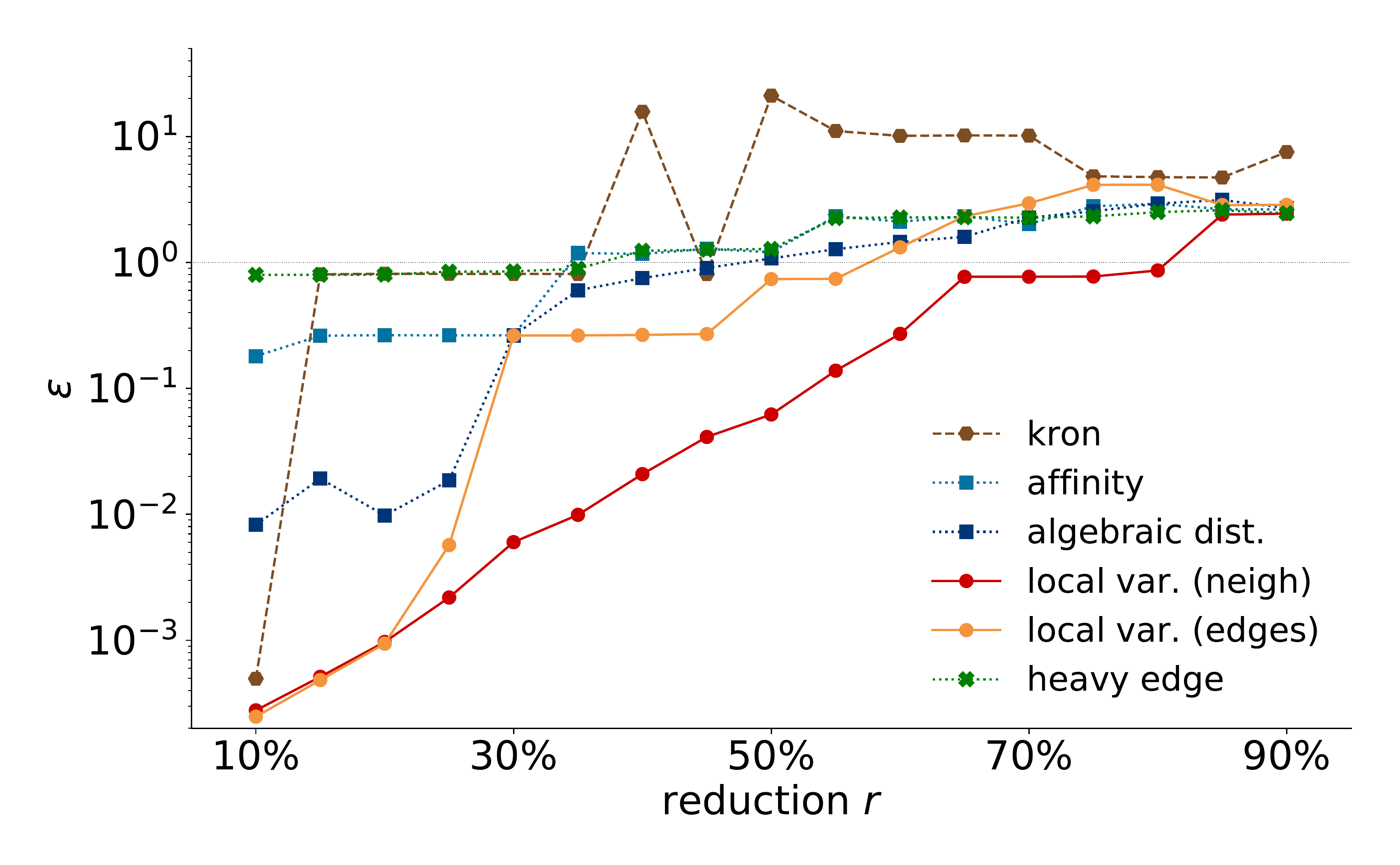}\label{fig:comp_yeast_10}}
  ~
  \subfloat[yeast ($k=40$)]{\includegraphics[width=\mys,trim={1.0cm 1cm 1.0cm 1cm},clip]{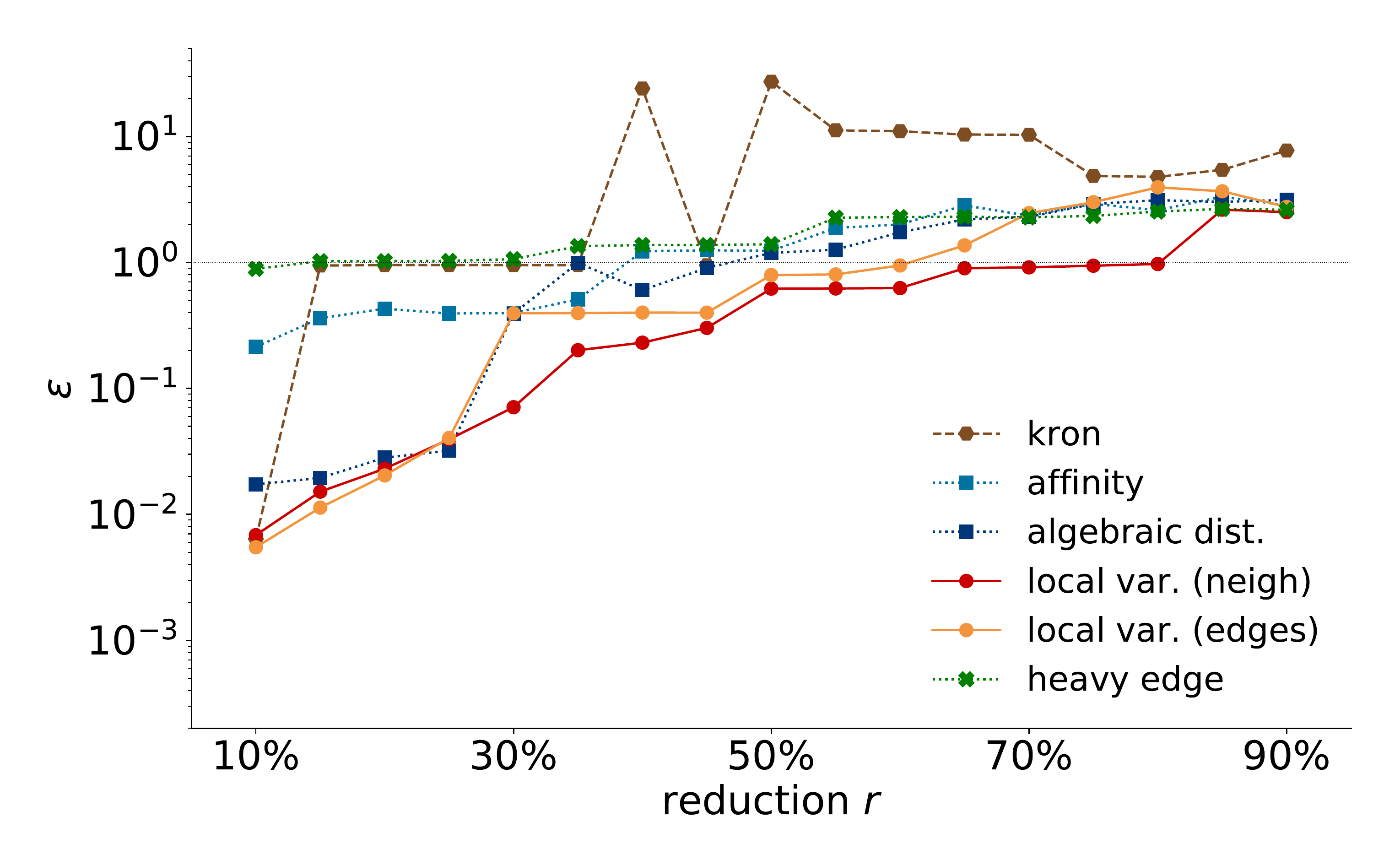}\label{fig:comp_yeast_40}} 
  \vspace{-2mm}\\
  \subfloat[airfoil ($k=10$)]{\includegraphics[width=\mys,trim={0.8cm 1cm 0.8cm 1cm},clip]{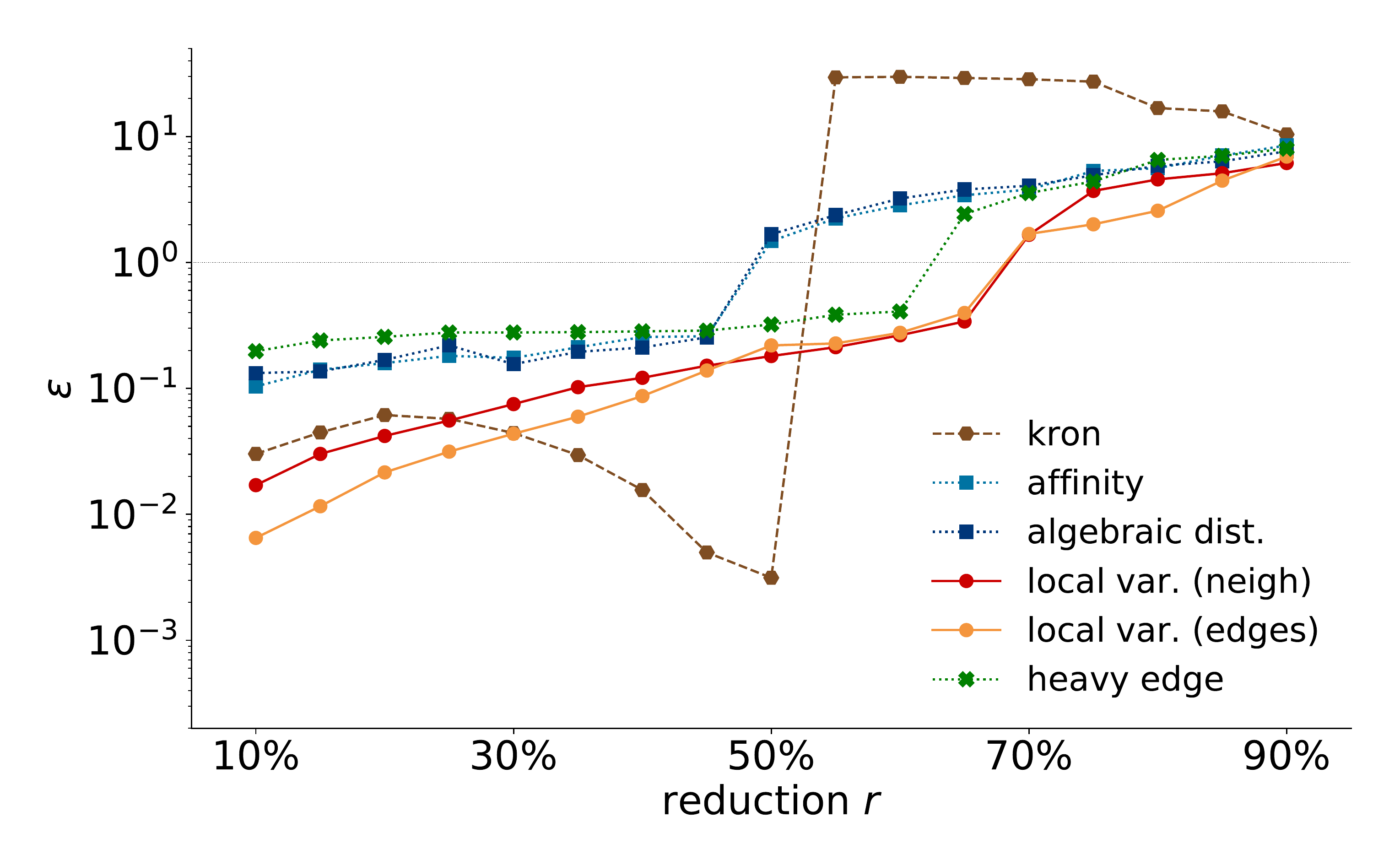}\label{fig:comp_airfoil_10}}
  ~
  \subfloat[airfoil ($k=40$)]{\includegraphics[width=\mys,trim={1.0cm 1cm 1.0cm 1cm},clip]{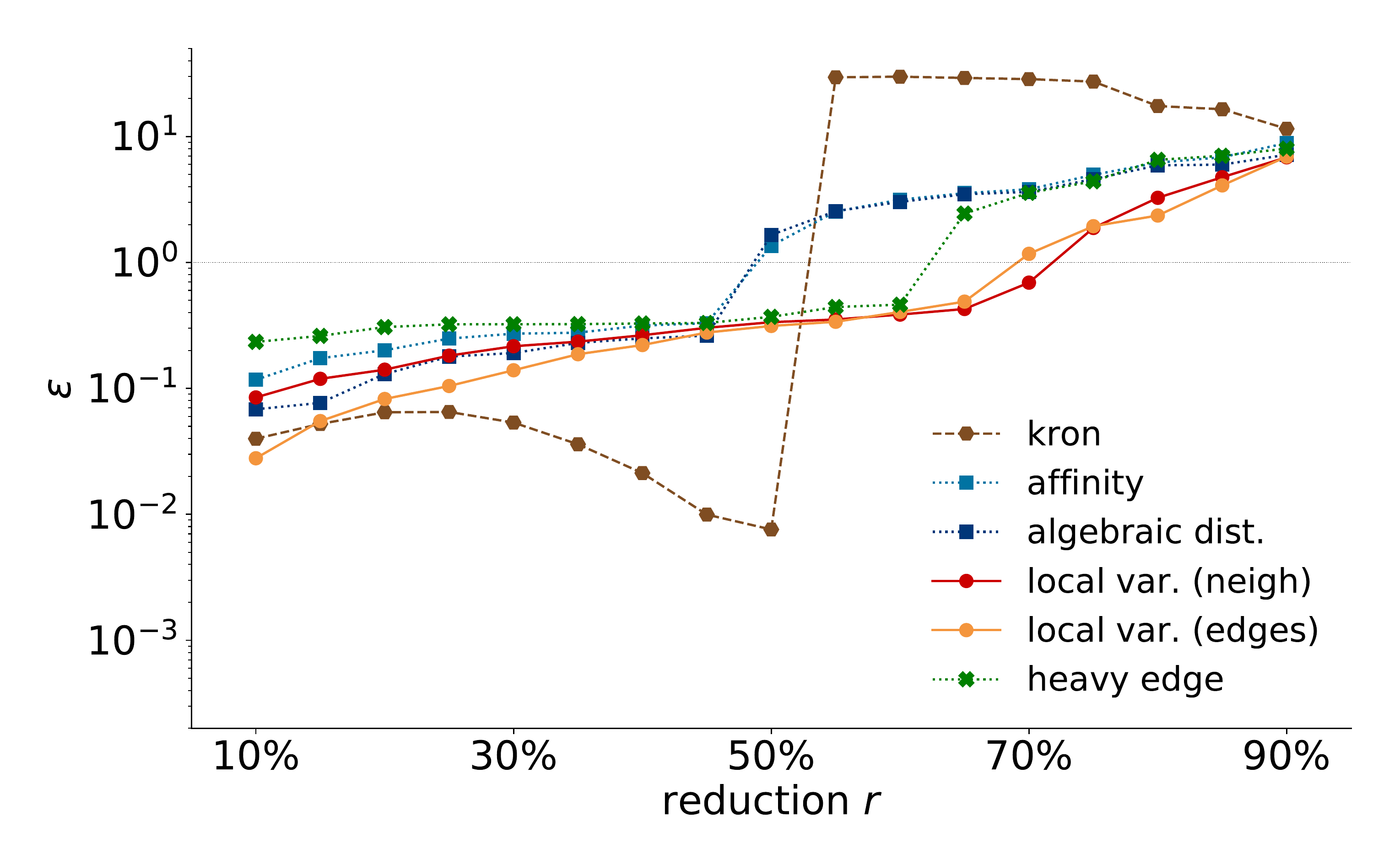}\label{fig:comp_airfoil_40}}
  \vspace{-2mm}\\
  \subfloat[bunny ($k=10$)]{\includegraphics[width=\mys,trim={0.8cm 1cm 0.8cm 1cm},clip]{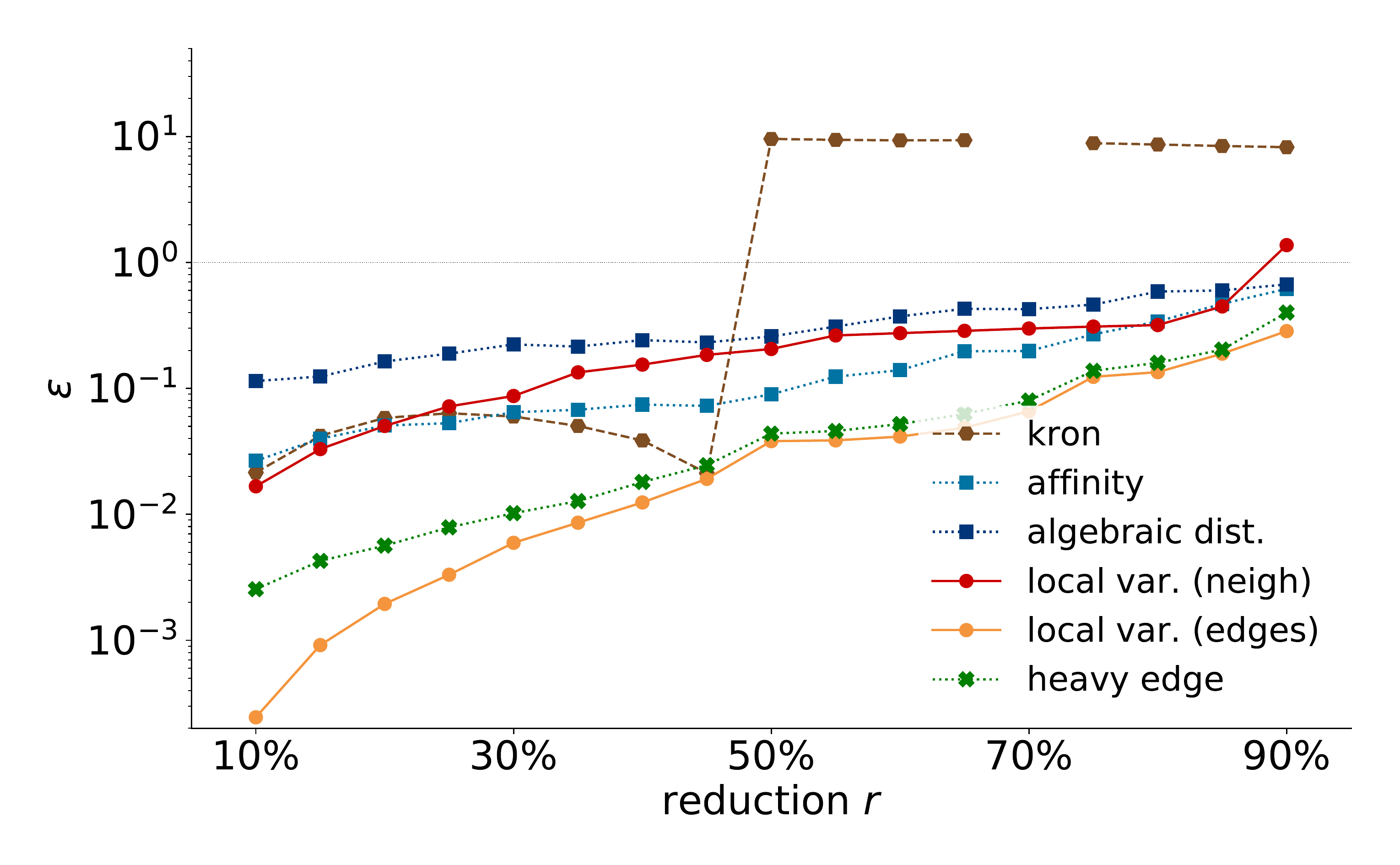}\label{fig:comp_bunny_10}}
  ~
  \subfloat[bunny ($k=40$)]{\includegraphics[width=\mys,trim={1.0cm 1cm 1.0cm 1cm},clip]{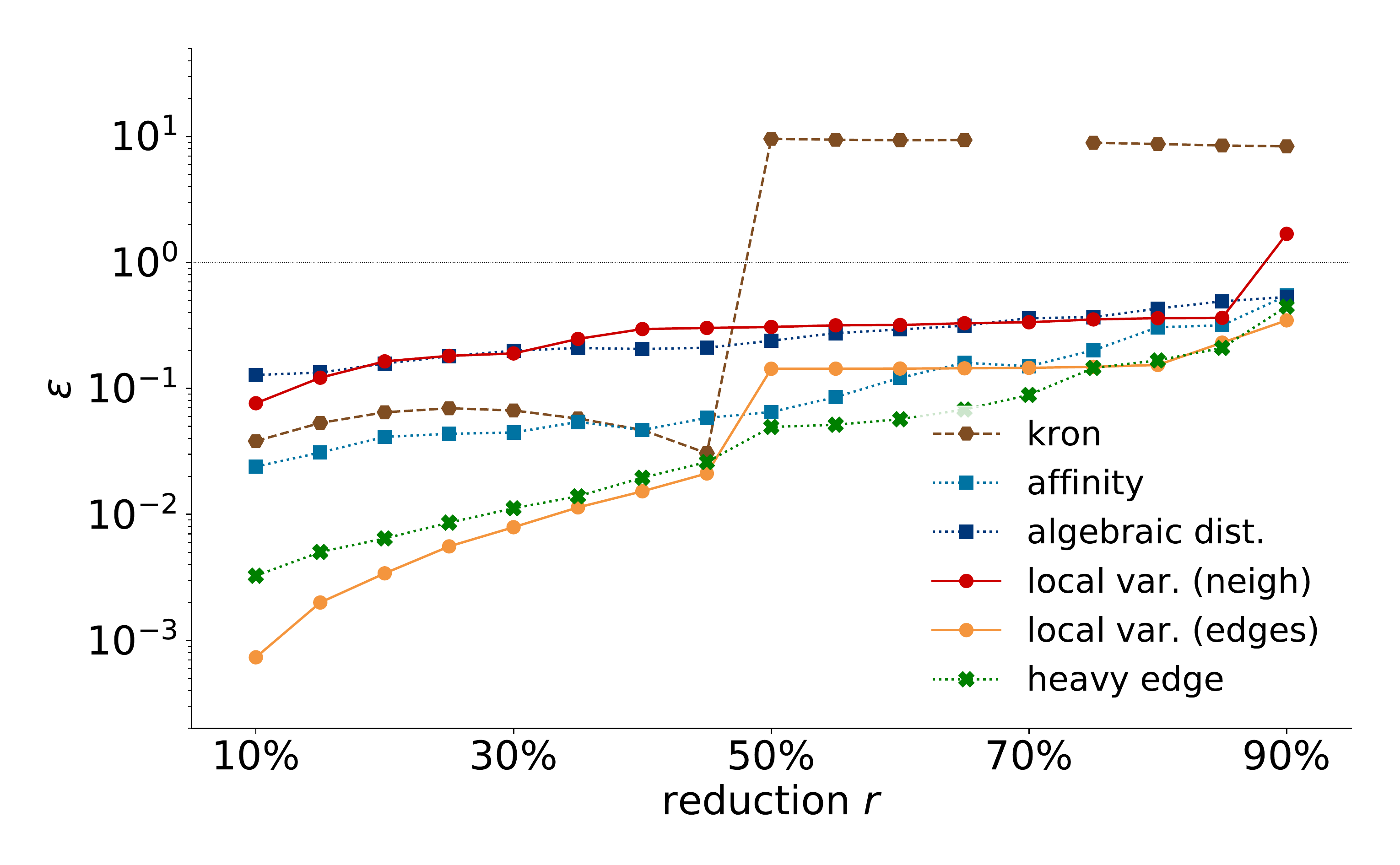}\label{fig:comp_bunny_40}}
  \vspace{-2mm}\\
  \subfloat[minnesota ($k=10$)]{\includegraphics[width=\mys,trim={0.8cm 1cm 0.8cm 1cm},clip]{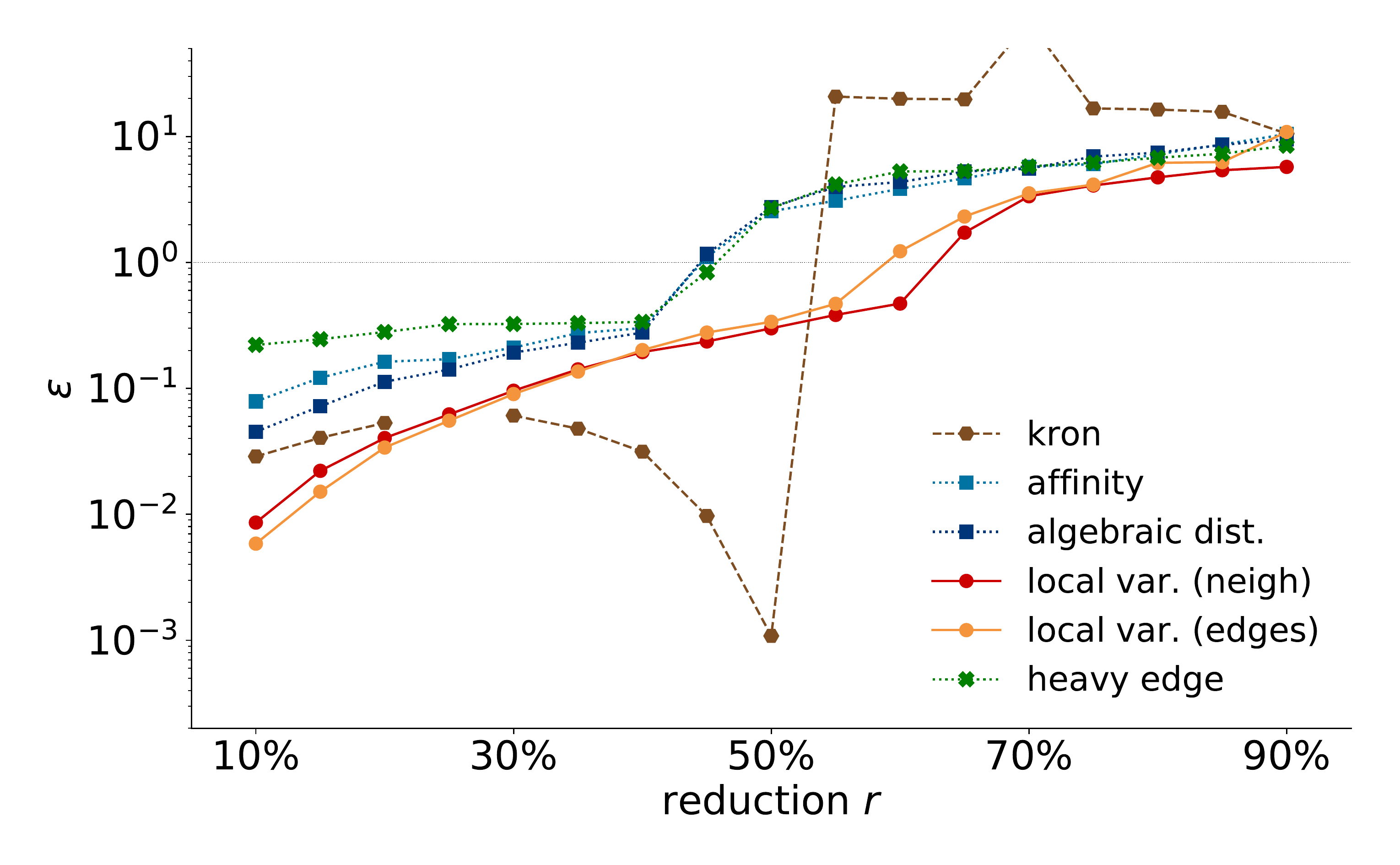}\label{fig:comp_minnesota_10}}
  ~
  \subfloat[minnesota ($k=40$)]{\includegraphics[width=\mys,trim={1.0cm 1cm 1.0cm 1cm},clip]{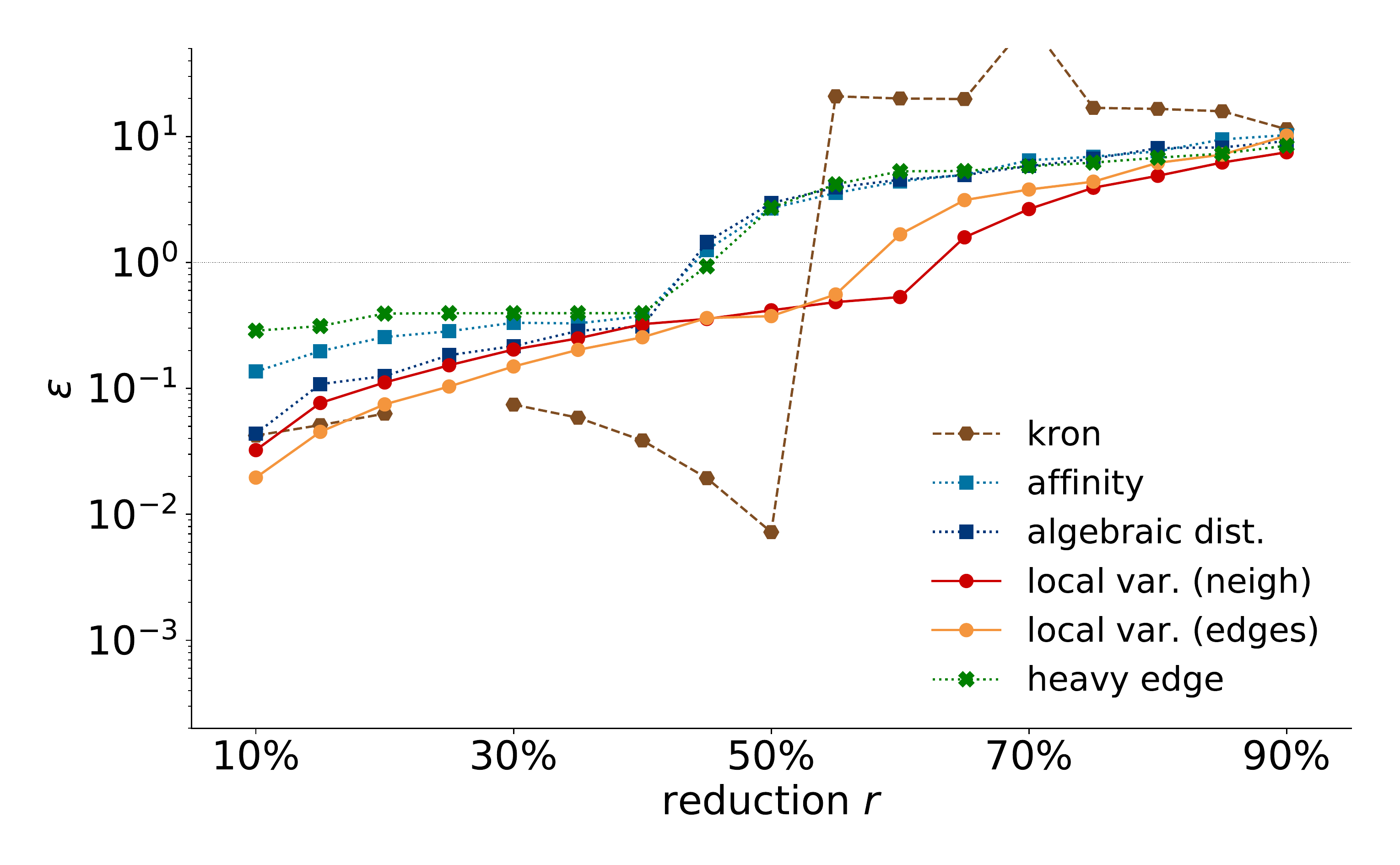}\label{fig:comp_minnesota_40}}
  \vspace{2mm}
  \caption{Quality of approximation comparison for four representative graphs (rows) and two subspace sizes (columns).
  \label{fig:comparison}}
\end{figure*}

The first experiment tests how well $L_c$ approximates the action of $L$ with respect to the subspace $U_k$ of smallest variation. In other words, for each method I plot the smallest $\epsilon$ such that the following equation holds: 
\begin{align}
    \norm{x - \p{x}}_{L} \leq \epsilon \, \norm{x}_L \quad \text{for all} \quad x \in \mathbf{U}_k.
\end{align}
The results are summarized in Figure~\ref{fig:comparison} for two representative subspaces of size $k=10$ and $k=40$. %
With the exception of the Kron reduction that was repeated 10 times, all methods are deterministic and thus were run only once. 

Overall, it can be seen that local variation methods outperform other coarsening methods in almost every problem instance. The gap is particularly prominent for large reductions, where multiple levels are employed. Neighborhood-based contraction yields the best result overall, mainly because it achieves the same reduction in fewer levels. 
Interestingly, local variation (and coarsening) methods in many cases also outperform Kron reduction, even though the latter is more demanding computationally. 

I elaborate further on four points stemming from the results:

\textit{In most instances, a reduction of up to 70\% is feasible while maintaining a decent approximation.} The attained approximation clearly is a function of the graph in question and $k$. Nevertheless, in almost all experiments, the best coarsening method could reduce the graph size by at least 70\% while also ensuring that $\epsilon<1$ (horizontal black line). This is an encouraging sign, illustrating that significant dimensionality reduction is often possible, without sacrificing too much the solution quality.   

\textit{Following intuition, it is generally harder to approximate subspaces of larger dimension $k$, but not excessively so}. Increasing $k$ from 10 to 40 in most cases increases $\epsilon$ only slightly. The only case where the approximation becomes profoundly better with small subspaces is with small reduction ratios $r$. For instance, coarsening the yeast graph results in an impressive approximation for all $r<30\%$ when $k=10$, whereas $\epsilon$ increases almost by an order of magnitude when $k$ becomes 40.  

\textit{Kron reduction is an effective way to half the graph size but can result in poor approximation otherwise.} If one is willing to sacrifice in terms of efficiency, Kron reduction effectively reduces the size of the graph by a factor of two (with the exception of the yeast graph). What might be startling is that the method behaves poorly for different reduction ratios. Three main factors cause this deterioration of performance. First, the sampling set is constructed based on the polarity of $u_N$ and has cardinality close to $N/2$~\citep{shuman2016multiscale}. Therefore, if in any level one tries to reduce the graph size by less than half, the last eigenvector heuristic cannot be used exactly. Second, numerical instability issues sometimes manifest when $r$ exceeds 50\%. Though I attempted to improve the original implementation featured in the PyGSP toolbox, some problem instances could not be solved successfully (hence the missing markers). The third reason is described next.

\textit{Coarse levels should aim to approximate the original graph and not the proceeding levels.} The conventional approach in multi-level schemes is to aim at each level to approximate as closely as possible the graph of the previous level. This can lead to sudden increase of error at consecutive levels (e.g. notice the minnesota error as $r$ approaches 50\%) as decisions early in the scheme can have a large impact later on. On the other hand, by Proposition~\ref{proposition:restricted_similarity} local variation methods modify the cost function minimized at each level, resulting in smoother transitions between levels and tighter approximations at large $r$.   


\subsection{Spectrum approximation}

The second part of the experiments examines the coarsening through the lens of spectral graph theory. 
The premise is that, since the spectrum of the Laplacian distills information about the graph structure, one may interpret the spectral distance as a proxy for the structural similarity of the two graphs. This is by no means a new idea---the Laplacian spectrum is a common ingredient in accessing graph similarity~\citep{wilson2008study}.

Tables~\ref{table:K=10} and~\ref{table:K=40} report the mean relative eigenvalue error defined as $\frac{1}{k} \sum_{i=1}^k \frac{|\p{\lambda_i} - \lambda_i|}{ \lambda_i} $ for two representative 
$k$, respectively 10 and 40. The results for $k=80$ were consistent with those presented here, and are not reported here for reasons of brevity.

\begin{table}[h!]
\footnotesize\centering
\resizebox{0.75\textwidth}{!}{
\begin{tabular}{@{}rccccccc@{}}
\toprule
                            & $r$                  & \begin{tabular}[c]{@{}c@{}}heavy\\ edge\end{tabular} & \begin{tabular}[c]{@{}c@{}}local var.\\ (edges)\end{tabular} & \begin{tabular}[c]{@{}c@{}}local var.\\ (neigh.)\end{tabular} & \begin{tabular}[c]{@{}c@{}}algebraic\\ distance\end{tabular} & affinity             & \begin{tabular}[c]{@{}c@{}}Kron\\ reduction\end{tabular}\\ \midrule
\multirow{3}{*}{yeast}      &                  30\% & 0.343                & 0.123                & \textbf{0.003}       & 0.145                & 0.177                & 0.054               \\
                            &                  50\% & 0.921                & 0.459                & \textbf{0.034}       & 0.605                & 1.020                & 1.321               \\
                            &                  70\% & 3.390                & 3.495                & \textbf{0.406}       & 3.504                & 3.732                & 1.865               \\\cmidrule(l){2-8}
\multirow{3}{*}{airfoil}    &                  30\% & 0.277                & \textbf{0.036}       & 0.065                & 0.213                & 0.245                & 0.352               \\
                            &                  50\% & 0.516                & 0.201                & \textbf{0.197}       & 1.268                & 1.375                & 0.912               \\
                            &                  70\% & 4.744                & 1.045                & \textbf{0.928}       & 5.544                & 5.775                & 1.984               \\\cmidrule(l){2-8}
\multirow{3}{*}{bunny}      &                  30\% & 0.019                & \textbf{0.006}       & 0.061                & 0.277                & 0.068                & 0.335               \\
                            &                  50\% & 0.064                & \textbf{0.046}       & 0.190                & 0.435                & 0.135                & 0.801               \\
                            &                  70\% & 0.126                & \textbf{0.081}       & 0.323                & 0.692                & 0.295                & 1.812               \\\cmidrule(l){2-8}
\multirow{3}{*}{minnesota}  &                  30\% & 0.332                & 0.088                & \textbf{0.078}       & 0.232                & 0.280                & 0.318               \\
                            &                  50\% & 2.018                & 0.432                & \textbf{0.310}       & 2.398                & 2.426                & 0.882               \\
                            &                  70\% & 9.299                & 4.579                & \textbf{1.866}       & 9.938                & 9.177                & 1.951               \\\bottomrule
\end{tabular}
}
\caption{Mean relative error for the first $k=10$ eigenvalues, for different graphs, reduction ratios, and coarsening methods.}
\label{table:K=10}
\end{table}


\begin{table}[h!]
\footnotesize\centering
\resizebox{0.75\textwidth}{!}{
\begin{tabular}{@{}rccccccc@{}}
\toprule
                            & $r$                  & \begin{tabular}[c]{@{}c@{}}heavy\\ edge\end{tabular} & \begin{tabular}[c]{@{}c@{}}local var.\\ (edges)\end{tabular} & \begin{tabular}[c]{@{}c@{}}local var.\\ (neigh.)\end{tabular} & \begin{tabular}[c]{@{}c@{}}algebraic\\ distance\end{tabular} & affinity             & \begin{tabular}[c]{@{}c@{}}Kron\\ reduction\end{tabular}\\ \midrule
\multirow{3}{*}{yeast}      &                  30\% & 0.328                & 0.113                & \textbf{0.023}       & 0.094                & 0.195                & 0.120               \\
                            &                  50\% & 0.879                & 0.430                & \textbf{0.130}       & 0.517                & 0.769                & 1.196               \\
                            &                  70\% & 2.498                & 2.182                & \textbf{0.451}       & 2.560                & 2.229                & 1.946               \\\cmidrule(l){2-8}
\multirow{3}{*}{airfoil}    &                  30\% & 0.277                & \textbf{0.095}       & 0.181                & 0.189                & 0.267                & 0.368               \\
                            &                  50\% & 0.549                & \textbf{0.325}       & 0.349                & 0.698                & 0.862                & 0.960               \\
                            &                  70\% & 2.268                & 0.872                & \textbf{0.839}       & 2.373                & 2.531                & 2.078               \\\cmidrule(l){2-8}
\multirow{3}{*}{bunny}      &                  30\% & 0.023                & \textbf{0.008}       & 0.085                & 0.205                & 0.052                & 0.294               \\
                            &                  50\% & 0.066                & \textbf{0.058}       & 0.181                & 0.346                & 0.089                & 0.660               \\
                            &                  70\% & 0.128                & \textbf{0.098}       & 0.299                & 0.509                & 0.202                & 1.192               \\ \cmidrule(l){2-8}
\multirow{3}{*}{minnesota}  &                  30\% & 0.353                & 0.118                & \textbf{0.115}       & 0.209                & 0.306                & 0.337               \\
                            &                  50\% & 1.259                & 0.468                & \textbf{0.383}       & 1.342                & 1.264                & 0.933               \\
                            &                  70\% & 4.162                & 2.111                & \textbf{1.612}       & 4.145                & 4.185                & 2.090               \\
\bottomrule
\end{tabular}
}
\caption{Mean relative error for the first $k=40$ eigenvalues, for different graphs, reduction ratios, and coarsening methods.}
\label{table:K=40}
\end{table}

As expected, the reduction ratio plays a major role in the closeness of Laplacian spectra. Indeed, for most cases the eigenvalue error jumps by almost an order of magnitude whenever $r$ increases by 20\%. Yet, in most cases acceptable errors can be achieved even when the coarse graph is as small as one third of the size of the original graph (corresponding to $r = 70\%$). 

It might also be interesting to observe that there is a general agreement between the trends reported here and those described in the matrix approximation experiment. In particular, if one sorts the tested methods from best to worse, he/she will obtain an ordering that is generally consistent across the two experiments, with local variation methods giving the best approximation by a significant margin. A case in point: for the maximum ratio, the best local variation method is on average 3.9$\times$ better than the leading state-of-the-art coarsening method. The gain is 2.6$\times$ if the Kron reduction is also included in the comparison.

Overall, it can be deduced from these results that local variation methods coarsen a graph in a manner that preserves its spectrum. This is in accordance with the theoretical results. As it was shown by Theorem~\ref{theorem:eigenvalues}, if $L$ and $L_c$ act similar w.r.t. all vectors in $\textbf{U}_k$, then their eigenvalues cannot be far apart. Therefore, by aiming for restricted approximation, local variation methods implicitly also guarantee spectrum approximation.

\subsection{Efficiency}

The last experiment tests computational efficiency. I adopt a simple approach and aim to coarsen a 10-regular graph of increasing size. I measure the execution time of the six different methods for graph reduction and report the mean over 10 iterations, while capping computation at 100 seconds. 

The results are displayed in Figures~\ref{fig:scalability_10} and~\ref{fig:scalability_40} for subspaces of size $10$ and 40, respectively. As with most such comparisons, the actual numbers are only indicative and depend on the programming language utilized (Python), processing paradigm (no parallelism was employed), and hardware architecture (2.2GHz CPU)\footnote{I expect that a significant speedup can be achieved by compiling the code to native machine instructions as well as by parallelizing the local variation cost computation.}.

Focusing on the trends, with the exception of Kron reduction and affinity, most methods scale quasi-linearly with the number of edges. Interestingly, local variation methods are quite competitive and do not sacrifice much as compared to the straightforward heavy edge matching. One can also observe that constructing $\F_{\ell-1}$ based on neighborhoods results in slightly slower computation that with the edge-based method construction, because in the latter case the local variation cost can be computed much more efficiently.

\begin{figure}[t]
  \centering
  \vspace{0mm}
  \subfloat[$k=10$]{\includegraphics[width=0.48\columnwidth,trim={0.8cm 1cm 0.8cm 1cm},clip]{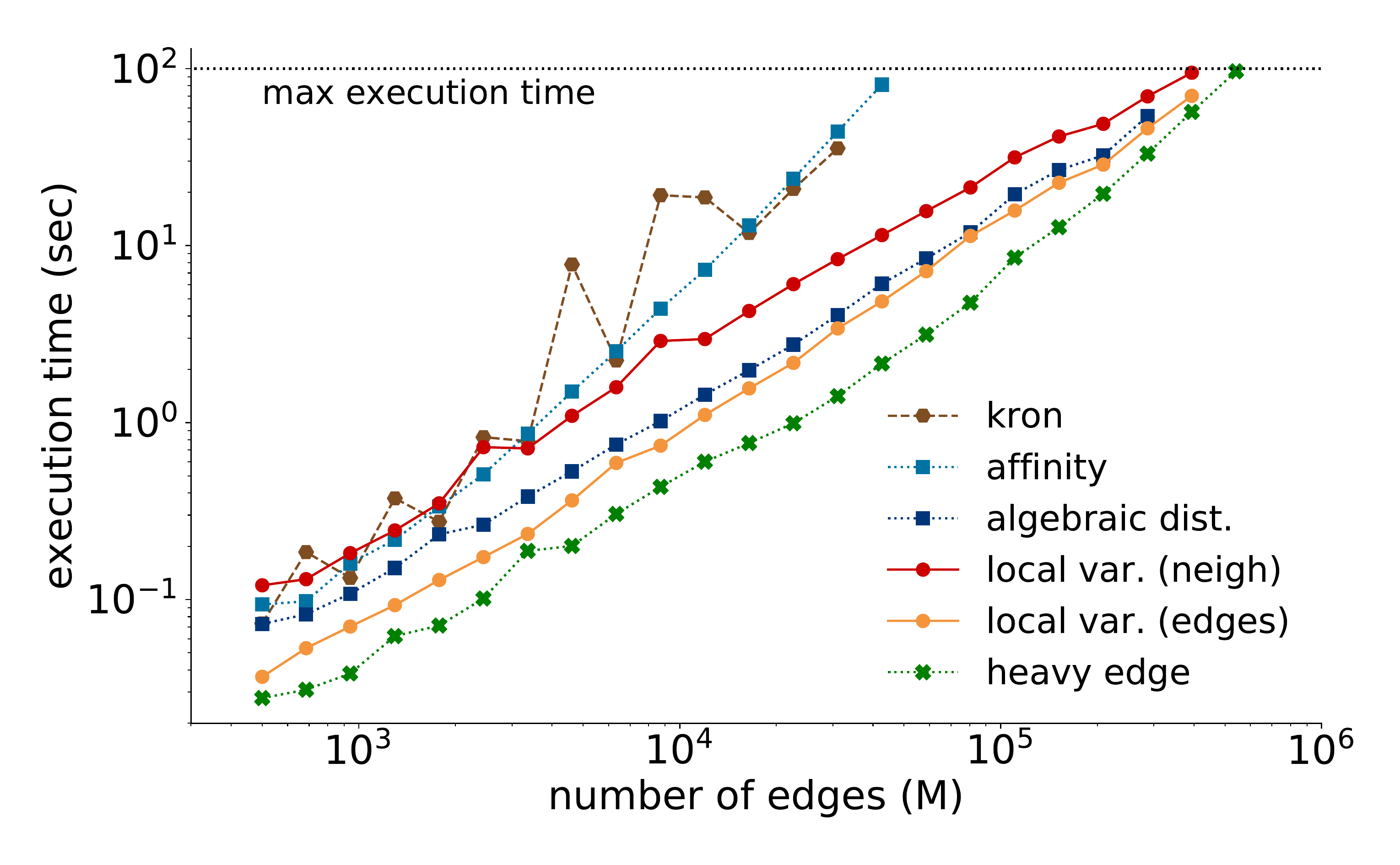}\label{fig:scalability_10}}
  ~
  \subfloat[$k=40$]{\includegraphics[width=0.48\columnwidth,trim={0.8cm 1cm 0.8cm 1cm},clip]{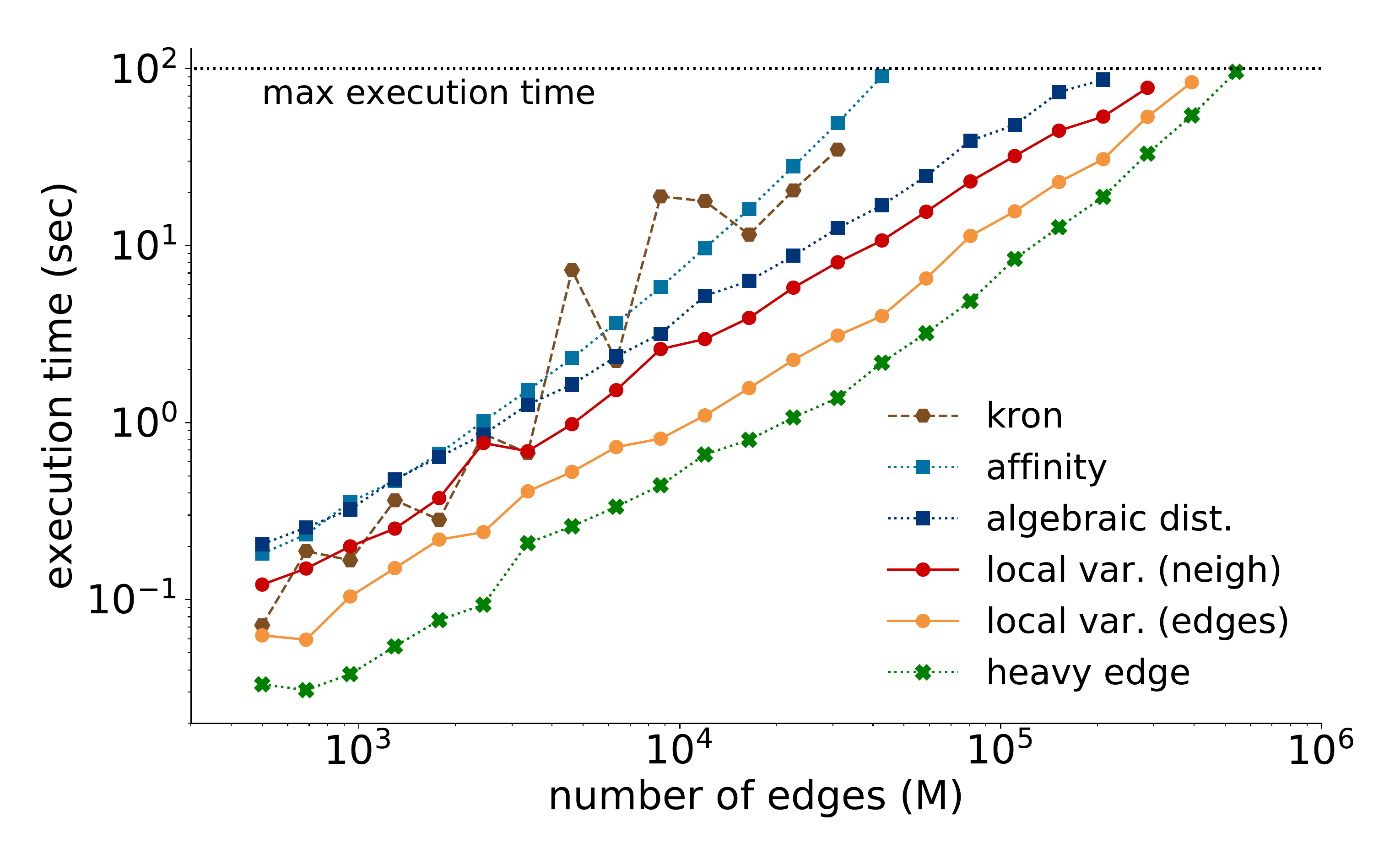}\label{fig:scalability_40}}
 \vspace{1mm}
\caption{Execution time as a function of the number of edges of the graph.\vspace{-3mm}}
\end{figure}

\section{Conclusions}

Graph reduction techniques are commonly used in modern graph processing and machine learning pipelines in order to facilitate the processing and analysis of large graphs. Nevertheless, more often than not these techniques have been designed based on intuition and possess no rigorous justification or provable guarantees.

This work considered the graph reduction problem from the perspective of restricted spectral approximation---a modification of the measure utilized for the construction of graph sparsifiers. This measure is especially relevant when restricted to Laplacian eigenspaces of small variation, as it implies strong spectral and cut guarantees.
The analysis of restricted spectral approximation has lead to new nearly-linear time algorithms for graph coarsening that provably approximate (a portion of) the graph spectrum, as well as the cut structure of a large graph.

A number of important questions remain open at the point of concluding this manuscript. To begin with, I am currently unaware of a rigorous way to determine how much one may benefit from reduction---that is, how small can $\epsilon$ be for a specific subspace and target $n$? In addition, no polynomial-time algorithm for graph coarsening exists that provably approximates the minimal achievable $\epsilon$. Finally, though I lack a formal proof, I also suspect that stronger cut guarantees can be derived from restricted spectral approximation. I argue that the potential of graph reduction cannot be fully realized until these fundamental questions are satisfactorily addressed. 

\paragraph{Acknowledgements.} This work was kindly supported by the Swiss National Science Foundation (SNSF) in the context of the project ``\emph{Deep Learning for Graph-structured Data}'', grant number PZ00P2$\_$179981. I would also like to thank Rodrigo C. G. Pena and Nguyen Q. Tran for their useful comments.

\bibliography{references}

\appendix
\input{appendix.tex}


\end{document}

%% file: appendix.tex
\section{Deferred proofs}

\subsection{Proof of Property~\ref{property:Pi_projection}}

%
I draw up an inductive argument demonstrating that $\Pi$ is a projection matrix.
The \textit{base case}, i.e., that $A_{c} = P_{c}^+ P_{c}$ is a projection matrix follows by the definition of the pseudo-inverse: $A_{c} A_{c} = P_{c}^+ P_{c} P_{c}^+ P_{c} = P_{c}^+ P_{c} = A_c$, where one uses the property $P_c = P_{c} P_{c}^+ P_{c}$. 

For the \textit{inductive step}, I argue that if 
$ A_{\ell+1} $ 
is a projection matrix the same holds for 
$$ A_{\ell} = P_{\ell}^+ \,  A_{\ell+1} \, P_{\ell}. $$
To this end, let $P_{\ell} = U \Sigma V^\top$ be the singular-value decomposition with $\Sigma = (D; 0) \in \Rbb^{N_{\ell+1} \times N_{\ell}}$ decomposed into the $N_{\ell+1} \times N_{\ell+1}$ diagonal matrix $D$ and the all zero matrix $0$. Then 
$$ A_{\ell} = V \Sigma^+ U^\top  A_{\ell+1} \, U \Sigma V^\top. $$
Recalling that a projection matrix remains projective if it undergoes a similarity transformation, we deduce that $U^\top  A_{\ell+1} \, U$ is also projective and, moreover, if $\Sigma^+ U^\top  A_{\ell+1} \, U \Sigma$ is a projection matrix, so is $A_{\ell}$.
However, one may write   
$$ 
\Sigma^+ U^\top  A_{\ell+1} \, U \Sigma = 
\left(\begin{array}{@{}cc@{}}
  D^{-1} & 0 \\
  0 & 0
\end{array}\right)
U^\top  A_{\ell+1} \, U
\left(\begin{array}{@{}cc@{}}
  D & 0 \\
  0 & 0
\end{array}\right)
= \left(\begin{array}{@{}cc@{}}
  D^{-1} U^\top  A_{\ell+1} \, U D & 0 \\
  0 & 0
\end{array}\right).
$$
As a block diagonal matrix whose blocks are projective (again $D^{-1} U^\top  A_{\ell+1} \, U D$ is a similarity transformation), $\Sigma^+ U^\top  A_{\ell+1} \, U \Sigma$ is also a projection matrix.
The proof that $\Pi$ is a projection matrix concludes by letting the induction unfold backwards from $c$ to 1 and setting $\Pi = A_1$.

\subsection{Proof of Proposition~\ref{proposition:pseudo-inverse}}

For notational convenience, I drop the level index supposing that $c=1$ and thus $P$ is an $n\times N$ coarsening matrix. As we will see in the following, $P$ has rank $n$ and thus to prove that $P^+ = {P}^\top D^{-2}$, it is sufficient to show that matrix $\Pi = {P}^\top D^{-2} P$ is a projection matrix of rank $n$ (and thus equal to $P^+ P$). 
Matrix $\tilde{P} = D^{-1} P$ has the same sparsity structure as $P$ and is thus also a coarsening matrix. 
W.l.o.g. let the rows of $P$ be sorted based on their support, such that for any two rows $r<r'$ and $P(r,i), P(r',i') \neq 0$ we necessarily have $i<i'$. 
Furthermore, denote by $p_r$ the vector containing all non-zero entries of $P(r,:)$ such that $\norm{p_r}_2 = \norm{P(r,:)}_2 = D(r,r)$. 
Due to the disjoint support of the rows of $P$ and under this particular sorting, matrix $\Pi$ is block-diagonal.
Moreover, each block $B_r$ in its diagonal is a rank 1 projection matrix as $B_r^2 = B_r B_r = \left(p_r D(r,r)^{-2} p_r^\top\right) \left(p_r D(r,r)^{-2} p_r^\top\right) = p_r D(r,r)^{-2} p_r^\top \frac{p_r^\top p_r}{\norm{p_r}_2^2} = B_r$. We have thus arrived to the relation $\Pi^2 = \Pi$, which constitutes a necessary and sufficient condition for $\Pi$ to be a projection matrix. The block-diagonal structure of $\Pi$ also implies that its rank (as well as that of $P$) is $n$.
%
%

\subsection{Proof of Proposition~\ref{proposition:proper_laplacian}}
%
Let us first remark that, by Proposition~\ref{proposition:pseudo-inverse}, $A = P^\mp$ is also a coarsening matrix with the same sparsity structure as $P$. 

\textit{Necessity.} I start by considering the nullspace of $\c{L} = A L A^\top$ and aim to ensure that it is equal to the span of the constant vector $1$, which is a necessity for all combinatorial Laplacian matrices. Since matrix $A$ is full row-rank and $L$ has rank $N-1$, the nullspace of $L_c$ is one dimensional. Therefore, the nullspace is correct as long as $ 1^\top A L A^\top 1 = 0$, which happens if either $A^\top 1 = \alpha 1$ for a constant $\alpha$ or $A^\top 1 = 0$. In both cases, $(A^\top 1)(r) = \alpha 1$ for every $r$. By the definition of $A$ however, we know that its rows have disjoint support and, as such, vector $A^\top 1$ exactly contains the non-zero entries of $A$. In other words, for the nullspace of $\c{L}$ to be properly formed, the non-zero entries of $A$ should either all be equal to $\alpha$ (such that $A^\top 1 = \alpha 1$) or zero (in which case $A^\top 1 = 0$). The latter case can clearly be discarded as it would disconnect the graph. 
We have thus discovered that $A L A^\top$ has a properly formed nullspace if and only if the non-zero entries of $A$ are equal, rendering the latter condition necessary. 

\textit{Sufficiency.} Considering that every Laplacian of $M$ edges can be re-written as $L = S^\top S$, where $S$ is the $M\times N$ incidence matrix, one may confirm that the condition is also sufficient by showing that, for every $A$ with equal non-zero entries, the matrix $\c{S} = S A^\top $ is an incidence matrix of $\c{L}$ such that $\c{L} = \c{S}^\top \c{S}$. W.l.o.g., suppose that $\alpha=1$ ($\alpha^2 L$ is a valid Laplacian for all $\alpha$). By construction, each row of $\c{S}$ is $\c{S}(q,:)^\top = A S(q,:)^\top$. Name as $e_{ij}$ the corresponding edge, such that $S(q,:)^\top = \delta_i - \delta_j$, where $\delta_i$ is a dirac centered at vertex $v_i$. It follows that $\c{S}(q,:)^\top = A \delta_i - A \delta_j $. Obviously, if none of $v_i,v_j$ are coarsened then $\c{S}(q,:)^\top = \delta_i - \delta_j$, which is a valid row. Moreover, by construction of $A$, if either of $v_i,v_j$ is coarsened (but not both) or if $v_i,v_j$ are coarsened into different vertices then both relations $A \delta_i = \delta_i$ and $A \delta_j = \delta_j$ hold and thus once more $\c{S}(q,:)^\top = \delta_i - \delta_j$ is a valid row. Lastly, if $v_i,v_j$ are coarsened into the same vertex then for some $r$ it must be that $A(r,i) = A(r,j)$, whereas $A(r',i) = A(r',j) = 0$ for all $r' \neq r$ and thus $\c{S}(r,:)^\top = 0$, signifying that the edge is not present.
Summarizing, in every case $\c{S}$ is a valid incidence matrix, rendering the condition also sufficient.

\subsection{Proof of Property~\ref{claim:cuts}}

For any two disjoint subsets $\S_1, \S_2$ of $\V$ denote by $w(\S_1, \S_2) = \sum_{v_i \in \S_1} \sum_{v_j \in \S_j} w_{ij}$ the cut weight in $G$. 

The claim is proven by induction on the number of levels. For the base case set $\ell = 1$ and define $C = P_1^+$ such that $c_r = C(:,r)$ is the indicator vector of the contraction set $\V_{0}^{(r)}$. 
It is a consequence of the Laplacian form of $L$ that, for any $v_r, v_q \in \V_1$ with $r\neq q$, we have 
\begin{align}
W_1(r,q) = -L_1(r,q) = -c_r^\top L c_q 
	&= \sum_{v_i \neq v_j}  w_{ij} c_r(i) c_q(j) + \sum_{v_i} d_i c_r(i) c_q(i) \notag \\
	&= \hspace{-4mm} \sum_{v_i \in \V_{0}^{(r)}, v_j \in \V_{0}^{(q)}} \hspace{-7mm} w_{ij} = w(\mathcal{S}_1^{(r)},  \mathcal{S}_1^{(q)}), \notag 
\end{align}
where the penultimate step uses $c_r(i)c_q(i) = 0$ since contraction steps are disjoint, and the last step exploits the equivalence $\V_{0}^{(q)} = \mathcal{S}_{1}^{(q)}$.
For the inductive step, consider level $\ell> 1$. Since $L_{\ell-1}$ is a Laplacian matrix, one may employ an identical argument as when $\ell=1$ to find that the weight between vertices $v_r,v_q \in \V_{\ell}$ with $r\neq q$ is  
$$ W_\ell(r,q) = \hspace{-4mm} \sum_{v_i \in \V_{\ell-1}^{(r)}, v_j \in \V_{\ell-1}^{(q)}} \hspace{-7mm} W_{\ell-1}(i,j).$$
By the induction hypothesis however, it must be $W_{\ell-1}(i,j) = w(\mathcal{S}_{\ell-1}^{(i)}, \mathcal{S}_{\ell-1}^{(j)})$, implying
$$
W_\ell(r,q) = \hspace{-4mm} \sum_{v_i \in \V_{\ell-1}^{(r)}, v_j \in \V_{\ell-1}^{(q)}} \hspace{-7mm} w(\mathcal{S}_{\ell-1}^{(i)}, \mathcal{S}_{\ell-1}^{(j)}) = w(\mathcal{S}_{\ell}^{(r)}, \mathcal{S}_{\ell}^{(q)}),$$
with the final equality being true due to the recursive definition of sets $\mathcal{S}_{\ell}^{(r)}$ and $\mathcal{S}_{\ell}^{(q)}$, as well as the following property of cuts:  for any two sets (call them large sets) and any partition of each into an arbitrary number of subsets, the cut between the large sets is equal to the sum of all cuts between pairs of subsets belonging to different large sets. 
This completes the proof.

\subsection{Proof of Theorem~\ref{theorem:interlacing}}

The Courant-Fischer min-max theorem for a Hermitian matrix $L$ reads  
\begin{align}
	\lambda_k = \min\limits_{\dimension{\mathbf{U}} = k}  \max\limits_{x \in \mathbf{U}} \left\{  \frac{x^\top L x}{x^\top x}   \, | \, x \neq 0  \right\}, \label{eq:courant_lambda_k}
\end{align}
whereas the same theorem for $\c{L}$ gives
\begin{align}
	\p{\lambda}_k = \min\limits_{\dimension{\c{\mathbf{U}}} = k}   \max\limits_{\c{x} \in \c{\mathbf{U}}} \ \left\{ \frac{\c{x}^\top \c{L} \c{x}}{\c{x}^\top \c{x}}  \, | \, \c{x} \neq 0 \right\} 
	&= \min\limits_{\dimension{\c{\mathbf{U}}} = k} \max\limits_{P x \in \c{\mathbf{U}}} \left\{ \frac{ x^\top \Pi L \Pi x}{x^\top P^\top P x}  \, | \, x \neq 0 \right\} , \notag 
\end{align}
where in the second equality I substitute $ \c{L} = P^\mp L P^+$ and $\c{x} = P x $.

We will need the following result: 
\begin{lemma}
	For any $P$ with full-row rank, the following holds: 
	$$ \lambda_{1}(P P^\top) \, x^\top \Pi x \leq x^\top P^\top P x \leq \lambda_{n}(P P^\top) \, x^\top \Pi x,$$
	with $\lambda_1(P P^\top)$ and $\lambda_n(P P^\top)$, respectively the smallest and largest eigenvalues of $P P^\top$.
	\label{lemma:projection}
\end{lemma}
\begin{proof}
Set $D = (P P^\top)^{+}$, which is an $n\times n$ PSD matrix. By the properties of the Moore–Penrose inverse $P^+ = P^\top  (P P^\top )^{+} = P^\top D $ and therefore 
$ P^\top P = P^\top D D^{-1} P = P^+ D^{-1} P $. Supposing that the eigenvalues of $D$ lie in $[\alpha, \beta]$
and that $P$ is full row-rank such that $\alpha> 0$, one may write
$$ \frac{1}{\beta} \, x^\top \Pi x \leq x^\top P^\top P x = x^\top P^+ D^{-1} P x \leq \frac{1}{\alpha} \, x^\top \Pi x.$$
To grasp why the aforementioned inequality holds, first use the cyclic property of the trace to obtain $$x^\top P^+ D^{-1} P x = \trace{x^\top P^+ D^{-1} P x} = \trace{D^{-1} P xx^\top P^+ },$$ and further recall that for any symmetric $A$ and PSD matrix $B$ one has $\lambda_{\text{min}}(A) \trace{B} \leq \trace{A B} \leq \lambda_{\text{max}}(B) \trace{B}$, where $\lambda_{\text{min}}(A)$ denotes the smallest eigenvalue of $A$ (and $\lambda_{\text{max}}(A)$ is the largest)~\citep{fang1994inequalities}. Set $A = D^{-1}$, which is by assumption symmetric, and $B = P xx^\top P^+$, which is PSD since its rank is (at most) one with the only (potentially) non-zero eigenvalue exactly $\trace{P xx^\top P^+} = x^\top P^+ P x = x^\top \Pi x = x^\top \Pi^2 x= \|\Pi x\|_2^2 \geq 0$. The desired inequality then follows since $ \lambda_{\text{min}}(D^{-1}) = 1/\lambda_{\text{max}}(D) = \frac{1}{\beta}$, $ \lambda_{\text{max}}(D^{-1}) = 1/\lambda_{\text{min}}(D) = \frac{1}{\alpha}$ and once more $\trace{P xx^\top P^+} = x^\top \Pi x$. 

Finally, since $P$ is full row-rank, $D$ is invertible meaning $\alpha = \lambda_{\text{min}}(D) = 1/ \lambda_{n}(P P^\top)$  and $\beta = \lambda_{\text{max}}(D) = 1/ \lambda_{1}(P P^\top)$.
\end{proof}

From the above, it is deduced that
\begin{align}	
	\p{\lambda}_k 
	&\geq \min\limits_{\dimension{\c{\mathbf{U}}} = k} \max\limits_{P x \in \c{\mathbf{U}}} \left\{ \frac{ x^\top \Pi L \Pi x}{ \lambda_{n}(P P^\top) \, x^\top \Pi x}  \, | \, x \neq 0 \right\}  \notag \\
	&\hspace{-4mm}= \frac{1}{\lambda_{n}(P P^\top)} \min_{\dimension{\mathbf{U}} = k, \mathbf{U} \subseteq \text{im}(\Pi) }  \max_{x \in \mathbf{U}} \left\{  \frac{x^\top L x}{x^\top x} \, | \, x \neq 0 \right\} , \notag 
\end{align}
where the equality holds since $\Pi$ is a projection matrix (see Property~\ref{property:Pi_projection}). Notice how, with the exception of the constraint $x = \Pi x$ and the multiplicative term, the final optimization problem is identical to the one for $\lambda_k$, given in~\eqref{eq:courant_lambda_k}. As such, the former's solution must be strictly larger (since it is a more constrained problem) and $\p{\lambda}_k \geq \frac{\lambda_k}{\lambda_{n}(P P^\top)} $.

Analogously, one obtains the lower inequality
$ \p{\lambda}_{k-(N-n)} \leq \frac{\lambda_k}{\lambda_1(PP^\top)}$
by applying the same argument on matrices $-L$ and $-\c{L}$ and exploiting that the $k$-th largest eigenvalue of any matrix $M$ is also the $k$-th smallest eigenvalue of $-M$.

\subsection{Proof of Theorem~\ref{theorem:eigenvalues}}

The lower bound is given by Theorem~\ref{theorem:interlacing}. For the upper bound, I reason similarly to the proof of the latter to find: 
\begin{align}
	\p{\lambda}_k &= \min\limits_{\dimension{\c{\mathbf{U}}} = k}   \max\limits_{\c{x} \in \c{\mathbf{U}}} \ \left\{ \frac{\c{x}^\top \c{L} \c{x}}{\c{x}^\top \c{x}}  \, | \, \c{x} \neq 0 \right\} \notag 
	&\leq \gamma_2 \min\limits_{\dimension{\mathbf{U}} = k, \textbf{U} \subseteq \image{\Pi}}  \max\limits_{x \in \mathbf{U}} \left\{  \frac{x^\top \Pi L \Pi x}{x^\top \Pi x} | \, x \neq 0  \right\}. \notag  
\end{align}
Above, the inequality is due to Lemma~\ref{lemma:projection} with $\gamma_2 = 1 / \lambda_{1}(PP^\top)$.
Thus, for any matrix $V$ the following inequality holds
\begin{align}
\p{\lambda}_k \leq \gamma_2 \max\limits_{x \in \spanning{V}} \left\{  \frac{x^\top \Pi  L \Pi x}{x^\top \Pi x}  \, | \, x \neq 0 \right\}, \notag 
\end{align}
as long as the image of $V$ is of dimension $k$ and does not intersect the nullspace of $\Pi$. 
Write $U_{k}$ to denote the $N \times k$ matrix with the $k$ first eigenvectors of $L$, whose image is of dimension $k$ as needed. Assume for now that the nullspace requirement is also met:
\begin{align}
\p{\lambda}_k 
&\leq \gamma_2 \max\limits_{x \in \spanning{U_k}} \left\{  \frac{x^\top \Pi L \Pi x}{x^\top \Pi x}  \, | \, x \neq 0 \right\} \notag 
= \gamma_2 \max\limits_{x \in \spanning{U_k}} \left\{  \frac{\|S \Pi x\|_2^2}{\|\Pi x\|_2^2 }  \, | \, x \neq 0 \right\}.
\end{align}
It will be convenient to manipulate the square-root of this quantity:
\begin{align}
\sqrt{\frac{\p{\lambda}_k}{\gamma_2}}
\leq \max\limits_{a \in \Rbb^{k}} \frac{\|S \Pi U_k a\|_2}{\| \Pi U_k a\|_2} 
= \frac{\|S \Pi U_k\|_2}{\|\Pi U_k\|_2} 
\leq \frac{\|S U_k\|_2 + \|S \Pi^\bot U_k\|_2}{\| \Pi U_k\|_2} 
&=\frac{\sqrt{\lambda_k} + \|S \Pi^\bot U_k\|_2}{\| \Pi U_k\|_2}, 
\label{eq:eigenvalue_intermediate_2}
\end{align}
with $S$ defined such that $L = S^\top S.$
The norm in the numerator is upper bounded by
\begin{align}
    \|S \Pi^\bot U_k\|_2 
    = \|S \Pi^\bot U_k \Lambda_k^{-\sfrac{1}{2}} \Lambda_k^{\sfrac{1}{2}}\|_2 
    &\leq \|S \Pi^\bot U_k \Lambda_k^{+\sfrac{1}{2}} \|_2 \|\Lambda_k^{\sfrac{1}{2}}\|_2 \notag \\
    &= \sqrt{\lambda_k} \, \|S \Pi^\bot U_k \Lambda_k^{+\sfrac{1}{2}} \|_2 = \sqrt{\lambda_k} \,\epsilon_k. \notag  
\end{align}
If the last step is not immediately obvious, one can be convinced by first exploiting the unitary-invariance of the spectral norm to write
$\|S \Pi^\bot U_k \Lambda_k^{+\sfrac{1}{2}} \|_2 = \|S \Pi^\bot U_k U_k^\top L^{+\sfrac{1}{2}} \|_2$,
and then confirming in the proof of Proposition~\ref{proposition:restricted_similarity} that the latter quantity is exactly $\epsilon_k$ when $V = U_k$.

The preceding analysis assumed that the image of $U_k$ and the nullspace of $\Pi$ did not intersect. Since $\Pi^\bot = I - \Pi$ is a complement projection matrix , the previous holds when $ \| \Pi^\bot U_k\|_2^2 < 1$. 
Since $\| \Pi^\bot u_1\| = 0$, one may w.l.o.g. exclude $u_1$ from the space of interest. For the remainder of $\image{U_k}$ the following holds:
$$\| \Pi^\bot U_k\|_2^2 = \max_{x \in \mathbf{U}_k \text{ and } x \bot u_1}  \frac{\|\Pi^\bot x\|_2^2 }{\|x\|_2^2} \leq \frac{1}{\lambda_2} \max_{x \in \mathbf{U}_k \text{ and } x \bot u_1} \frac{\| \Pi^\bot x \|_L^2}{\|x\|_2^2} = \epsilon_k^2 \frac{\lambda_k}{\lambda_2}.
$$
Therefore, when $\epsilon_k^2 < \frac{\lambda_2}{\lambda_k}$, the nullspace condition is met.
The proof is then concluded by substituting the bound $ \|\Pi U_k\|_2^2 = { 1- \|\Pi^\bot U_k\|_2^2} \geq { 1- \epsilon_k^2 \frac{\lambda_k}{\lambda_2}} $ in the denominator of~\eqref{eq:eigenvalue_intermediate_2}.

\subsection{Proof of Theorem~\ref{theorem:sintheta}}

\citet{li1994relative} showed that we can express the sin$\Theta$ as a sum of squared inner products: 
\begin{align}
	\norm{ \sintheta{U_k}{P^\top \p{U}_k} }_F^2 
	&= \norm{\p{U}_{k^\bot}^\top P U_k}_F^2 
	= \sum_{i \leq k} \sum_{j > k} (\p{u}_j^\top P u_i)^2
\label{eq:sintheta_0}
\end{align}
If $L_c$ and $L$ are $(\mathbf{U}_i, \epsilon_i)$-similar it follows from Corollary~\ref{corollary:restricted_isometry} that 
$$u_i^\top P^\top \c{L} P u_i \leq (1+\epsilon_i)^2 \lambda_i. $$
%
Summing these inequalities for all $i \leq k$ amounts to
\begin{align}
	\sum_{i \leq k} (1 + \epsilon_i)^2\lambda_i 
	\geq 
	&\sum_{i \leq k} \sum_{j = 1}^n \p{\lambda}_j (\p{u}_j^\top P u_i)^2 \notag \\
	&\geq \gamma_1 \sum_{i \leq k} \sum_{j = 1}^n \lambda_j (\p{u}_j^\top P u_i)^2 \notag \\
	&= \gamma_1 \sum_{j \leq k} \lambda \sum_{i \leq k} (\p{u}_j^\top P u_i)^2 + \gamma_1 \sum_{j > k} \lambda_j \sum_{i \leq k} (\p{u}_j^\top P u_i)^2.
\label{eq:sintheta_1}
\end{align}
where following Theorem~\ref{theorem:interlacing} I set $\gamma_1 = 1/\lambda_n(P P^\top)$, such that $\tilde{\lambda}_i \geq \gamma_1 \lambda_i $. Perform the following manipulation: 
\begin{align}
	\sum_{j \leq k} \lambda_j \sum_{i \leq k} (\p{u}_j^\top P u_i)^2 
	\geq \sum_{j \leq k} \lambda_j \sum_{i \leq k} (\p{u}_j^\top P u_i)^2  
	&= \sum_{j \leq k} \lambda_j \left( 1 - \sum_{i > k} (\p{u}_j^\top P u_i)^2 \right) \notag \\
	&\hspace{-10mm}\geq \sum_{j \leq k} \lambda_j - \lambda_k  \sum_{i \leq k} \left( \| \Pi^\bot u_i\|_2^2 + \sum_{j \geq k} (\p{u}_j^\top P u_i)^2 \right),
	\notag 
\end{align}
which together with~\eqref{eq:sintheta_0} and~\eqref{eq:sintheta_1} yields
\begin{align}
	\norm{ \sintheta{U_k}{P^\top \p{U}_k} }_F^2 
	&\leq \sum\limits_{i \leq k}  \frac{  (1 + \epsilon_i)^2\lambda_i/\gamma_1  - \lambda_i  + \lambda_k \| \Pi^\bot u_i \|_2^2 }{\lambda_{k+1} - \lambda_{k}}  
\end{align}
%
%
To proceed, I note the following useful inequality:
\begin{lemma}
	If $L$ and $L_c$ are $(\mathbf{U}_i,\epsilon_i)$-similar, then $\| \Pi^\bot u_i\|_2^2 \leq \epsilon_i$.
	\label{lemma:pinorm}
\end{lemma}
\begin{proof}
For every $u_i$ one sees that
\begin{align}
	u_i^\top P^\top \c{L} P u_i = u_i^\top \Pi L  \Pi u_i &= u_i^\top (I - \Pi^\bot)  L (I - \Pi^\bot) u_i \notag \\
	&= u_i^\top L u_i - 2 u_i^\top L \Pi^\bot u_i + u_i^\top \Pi^\bot L \Pi^\bot u_i \notag \\
	&= \lambda_i - 2 \lambda_i u_i^\top \Pi^\bot u_i + u_i^\top \Pi^\bot L \Pi^\bot u_i \notag 	
\end{align}
meaning that 
\begin{align*}
	\| \Pi^\bot u_i\|_2^2 = \frac{1}{2} \left( 1 + \frac{u_i^\top \Pi^\bot L \Pi^\bot u_i}{\lambda_i} - \frac{u_i^\top P^\top \c{L} P u_i}{\lambda_i}\right) \leq \frac{1}{2} \left( 1 + \epsilon_i^2 - (1 - \epsilon_i)^2  \right) = \epsilon_i.
\end{align*}
The last inequality is because, by restricted spectral approximation, we have $u_i^\top \Pi^\bot L \Pi^\bot u_i = \|\Pi^\bot u_i\|_L^2 \leq \epsilon_i^2 \|u_i\|_L^2 = \epsilon_i^2 \lambda_i$ and from Corollary~\ref{corollary:restricted_isometry}: $$u_i^\top P^\top \c{L} P u_i = \|P u_i\|_{L_c}^2 \geq (1 - \epsilon_i)^2 \|u_i\|_L^2 = (1 - \epsilon_i)^2 \lambda_i.$$
\end{proof}
As a consequence, it follows that
\begin{align}
	\norm{ \sintheta{U_k}{P^\top \p{U}_k} }_F^2 
	&\leq \sum\limits_{i \leq k}  \frac{  (1 + \epsilon_i)^2\lambda_i/\gamma_1  - \lambda_i  + \lambda_k \epsilon_i }{\lambda_{k+1} - \lambda_{k}}, \notag  
\end{align}
which after manipulation gives the desired inequality.

\subsection{Proof of Theorem~\ref{theorem:cheeger}}

The lower bound is a direct consequence of consistent coarsening and holds independently of restricted spectral approximation: for any set $\mathcal{S}_c \subset \mathcal{V}_c$ define $\mathcal{S} \subset \mathcal{V}$ such that $v_i \in \mathcal{S}$ if and only if $\varphi_c \circ \varphi_{c-1} \circ \cdots \varphi_1(v_i) \in \mathcal{S}_c$. Clearly, $w(\mathcal{S}) \geq w_c(\mathcal{S}_c)$, where the subscript $c$ implies that the latter volume is w.r.t. $G_c$. In addition, by the definition of Laplacian consistent coarsening and since every contraction set belongs either in $\mathcal{S}$ or $\bar{\mathcal{S}}$ (but not in both), it follows that $w(\mathcal{S}, \bar{\mathcal{S}}) = w_c(\mathcal{S}_c, \bar{\mathcal{S}}_c)$. In other words, for every $\mathcal{S}_c$ there exists a set $\mathcal{S}$ such that $$\phi(\mathcal{S}) = \frac{ w(\mathcal{S}, \bar{\mathcal{S}}) }{ \min\{w(\mathcal{S}), w(\bar{\mathcal{S}})\} } \leq \frac{ w_c(\mathcal{S}_c, \bar{\mathcal{S}_c}) }{ \min\{w_c(\mathcal{S}_c), w_c(\bar{\mathcal{S}}_c)\} } = \phi_c(\mathcal{S}_c),$$ implying also that the $k$-conductance of $G$ and $G_c$ are related by $\phi_k(G) \leq \phi_k(G_c)$. 

For the upper bound, I exploit the following multi-way Cheeger inequality:
\begin{theorem}[Restatement of Theorem 1.2 by~\cite{lee2014multiway}]
For every graph $G$ and every $k \in \mathbb{N}$, we have 
$$
\frac{\mu_{k}}{2} \leq \phi_{k}(G) = O(\sqrt{\mu_{2k} \, \xi_k(G)}), 
$$
with $\xi_k(G) = \log{k}$. If $G$ is planar then $\xi_{k}(G) = 1.$ More generally, if $G$ excludes $K_h$ as a minor, then $\xi_{k}(G) = h^4.$
\end{theorem}
Further, in the standard Cheeger inequality~\citep{alon1985lambda1,alon1986eigenvalues} ($k=2$) the upper bound is $\sqrt{2 \mu_2}$. Note that the eigenvalues mentioned here are those of the normalized Laplacian matrix $L^n = D^{-\sfrac{1}{2}} L D^{-\sfrac{1}{2}}$. To this end, suppose that $V_{2k}$ contains the first ${2k}$ eigenvectors of $L^n$ and fix $\mathbf{R} = \spanning{D^{-\sfrac{1}{2}} V_{2k}}$. Perform consistent coarsening w.r.t. to the combinatorial Laplacian $L$ (not $L^n$). Then, by definition, if $L_c$ and $L$ are $(\epsilon_{2k}, \mathbf{R})$-similar then for every $x \in \mathbf{R}$ one gets
\begin{align*}
	\|\Pi^\bot x\|_{L} \leq \epsilon_{2k} \|x\|_L. 
\end{align*} 
The substitution $y = D^{\sfrac{1}{2}} x$, such that $y \in \spanning{V_{2k}}$, allows us to transform the semi-norms above into semi-norms concerning $L^n$ as follows:
$$ 
\|x\|_L^2 = x^\top L x = y^\top D^{-\sfrac{1}{2}} L D^{-\sfrac{1}{2}} y = \|y\|_{L^n}^2
$$
and
$$ 
\|\Pi^\bot x\|_{L} = \|D^{-\sfrac{1}{2}} D^{\sfrac{1}{2}} \Pi^\bot D^{-\sfrac{1}{2}} D^{\sfrac{1}{2}} x\|_{L} = \|D^{\sfrac{1}{2}} \Pi^\bot D^{-\sfrac{1}{2}} y\|_{L^{n}} = \|(I - \Pi^{n}) y\|_{L^{n}}. 
$$
Above, $\Pi^n = D^{\sfrac{1}{2}} \Pi D^{-\sfrac{1}{2}} $ is the projection matrix (the set of projection matrices is closed under similarity transformations) corresponding to the coarsening matrix $P^n = P D^{-\sfrac{1}{2}}$, and now $y_c = \Pi^n y$.
It follows that, for every $y \in \mathbf{V}_{2k} = \spanning{V_{2k}}$, we have
\begin{align*}
	\|y - \Pi^n y\|_{L^n} \leq \epsilon_{2k} \|y\|_{L^n}
\end{align*} 
and thus $L_c^n$ and $L^n$ are $(\mathbf{V}_{2k}, \epsilon_{2k})$-similar.

Combining the multi-way Cheeger inequality with Theorem~\ref{theorem:eigenvalues} for $L_c^n$ and $L^n$ one obtains
\begin{align*}
\phi_{k}^2(G_c) &= O \left( \tilde{\mu}_{2k} \, \xi_k(G)\right)  \\
	&= O\left( \gamma_2 (1 + \epsilon_{2k})^2 \frac{ \mu_2 }{\mu_2 - \epsilon_{2k} \mu_{2k} } \mu_{2k} \, \xi_{k}(G) \right) \\
	&= O\left( \frac{ \gamma_2 \, (1 + \epsilon_{2k})^2 \xi_{k}(G) }{1 - \epsilon_{2k}^2 (\mu_{2k}/\mu_2)} \phi_k(G) \right),
\end{align*}
where the eigenvalues above are those of $L^n$ and the preceeding holds whenever $\epsilon_{2k}^2 < \mu_{2}/\mu_{2k}$. Further, when $k=2$ the upper bound simplifies to
$
\phi_{2}^2(G_c) \leq \frac{ 4 \, \gamma_2 \, (1 + \epsilon_{2})^2 }{1 - \epsilon_{2}^2} \phi_2(G).
$

\subsection{Proof of Proposition~\ref{proposition:lower_bound}}

Write $\mathcal{P} = \{\S_1, \S_2, \ldots, \S_n\}$ to denote a partitioning of $\V$ into $n$ disjoint clusters. In matrix form, the $n$-means problem performed on $U_k$ corresponds to finding a map $\varphi : \V \rightarrow \mathcal{P}$ between points $x_i = U_k(i,:)$ and clusters $\S_1, \S_2, \ldots, \S_n$ that minimizes the following cost function:
$$ \kmeans{n}{U_k}{\mathcal{P}} = \sum_{i = 1}^N \left( x_i - \hspace{-2mm}\sum_{v_j \in \varphi(v_i)} \frac{x_j}{|\varphi(v_i)|} \right)^2 = \| U_k - C^\top C U_k\|_F^2,$$
where the cluster indicator matrix $C \in \Rbb^{n \times N}$ has as entries 
$$
	C(r,j) = 
	\begin{cases}
		1/\sqrt{|\S_r|} & \text{if } v_i \in \S_r \\
		0 			 & \text{otherwise.}
	\end{cases}
$$
It may be confirmed that $C$ is a proper coarsening matrix corresponding to the case that clusters are exactly contraction sets. In addition, since $C^\top = C^+$ the matrix $C^\top C$ is the familiar projection matrix associated with coarsening. In fact, the latter is exactly equivalent to the projection matrix $\Pi = P^+ P$ of a single level Laplacian consistent coarsening. To see this, construct the diagonal matrix $Q$ with $Q(r,r) = \sqrt{|\S_r|}$ and write $C^\top C = C^\top Q Q^{-1} C = P^+ P = \Pi$.
With this in place, we can re-write the $n$-means problem as 
$$
	\kmeans{n}{U_k}{\mathcal{P}} = \| U_k - \Pi U_k\|_F^2 = \| \Pi^\bot U_k\|_F^2 = \sum_{i\leq k} \| \Pi^\bot u_i\|_2^2 \leq \sum_{i\leq k} \epsilon_i, 
$$
where $\epsilon_i$ is the smallest constant such that $L_c = P^\pm L P^+$ and $L$ are $(\mathbf{U}_i, \epsilon_i)$-similar and the inequality follows from Lemma~\ref{lemma:pinorm}. 
The final lower bound is then achieved by minimizing over all maps $\varphi$.

\subsection{Proof of Proposition~\ref{proposition:restricted_similarity}}

%
The following analysis is slightly more general than what is claimed in the statement of Proposition~\ref{proposition:restricted_similarity}: it holds for arbitrary PSD $L$ and $L_c$ (i.e., not necessarily Laplacian matrices) as long as the image $\image{\Pi}$ of the projection matrix $\Pi$ encloses the nullspace of $L$. The former trivially holds for Laplacian consistent coarsening, as, by design, one has $\Pi 1 = 1$ (see Section~\ref{subsec:consistent_coarsening}).

Let $V \in \Rbb^{N\times k}$ be a basis of $\mathbf{R}$. I start by proving that, for any integer $k\leq n$ and for all $x \in \spanning{V}$ the inequality
\begin{align}
    \|x - \Pi x\|_{L} \leq \epsilon \, \norm{x}_L \notag
\end{align}
holds for all $\epsilon \geq \|\Pi^\bot B_0\|_L$, where $B_0 = V V^\top L^{+\sfrac{1}{2}}$. I remind the reader that $\|x\|_L = \|S x \|_2 = \|L^{\sfrac{1}{2}} x\|_2$ and $\Pi^\bot = I - P^+P$. 
Furthermore, since $\image{\Pi}$ necessarily encloses the nullspace $\mathbf{N}$ of $L$, w.l.o.g., one may assume that $ \forall x \in \mathbf{R}$ matrix $L$ is invertible. To see why, note that if $x\in \mathbf{N}$ then $\|x\|_L = 0$ and $\| x - \Pi x \|_L = 0$, meaning that the inequality above is trivially satisfied.
I then derive 
\begin{align}
 \max\limits_{x \in \mathbf{R}} \frac{\|x - \Pi x\|_L}{\|x\|_L } 
 	&= \max\limits_{x \in \mathbf{R}} \frac{\| S \Pi^\bot x\|_{2}}{\|L^{\sfrac{1}{2}} x\|_2 }  \notag \\
 	&= \max\limits_{x \in \mathbf{R}} \frac{\| S \Pi^\bot V V^\top x\|_{2}}{\|L^{\sfrac{1}{2}} x\|_2 } \label{eq:aa1} \\
 	&= \max\limits_{x \in \image{L V}} \frac{\| S \Pi^\bot V V^\top L^{+\sfrac{1}{2}} x\|_{2}}{\|x\|_2 } \label{eq:aa2} \\
	&\leq \|S \Pi^\bot V V^\top L^{+\sfrac{1}{2}}\|_2 = \|\Pi^\bot B_0\|_L, \notag  
\end{align}
where equality~\eqref{eq:aa1} holds because $VV^\top$ is a projection onto $\mathbf{R}$, whereas equality~\eqref{eq:aa2} is true since $L$ is invertible within $\mathbf{R}$.

One should also note that, for the specific case where $V$ is an eigenspace of $L$, $\image{LV} = \textbf{R}$ and as such $\epsilon = \|\Pi^\bot x\|_{L} / \norm{x}_L = \|S \Pi^\bot V V^\top L^{+\sfrac{1}{2}}\|_2$ (once more w.l.o.g. the nullspace of $L$ can be ignored).

In addition, as the following technical lemma claims, in a multi-level scheme, any $\|\Pi^\bot x\|_L$ can be broken down into the contributions of each level:

\begin{lemma}
Define projection matrices $\Pi_\ell = P_{\ell}^+ P_\ell$ and $\Pi_\ell^\bot = I - \Pi_\ell$. If 
\begin{align}
    \| \Pi^\bot_\ell x_{\ell-1} \|_{L_{\ell-1}} \leq \sigma_\ell \, \| x_{\ell-1} \|_{L_{\ell-1}} 
    \quad \text{at each level } \ell \leq c, \notag
\end{align}
then the multi-level error is bounded by
\begin{align}
    \|\Pi^\bot x\|_L \leq \left( \sum_{\ell= 1}^c \sigma_\ell \prod_{q = 1}^{\ell-1} ( 1 + \sigma_q) \right) \|x\|_L = \left(\prod_{\ell=1}^{c} ( 1 + \sigma_\ell) - 1\right) \|x\|_L.\notag
\end{align}    
\label{lemma:leveling}
\end{lemma}
\begin{proof}
Recursively apply the following inequality
\begin{align}
    \norm{S_{\ell-1} \Pi^\bot_\ell x_{\ell-1}}_2 
    &\leq \sigma_\ell \norm{S_{\ell-1} x_{\ell-1}}_2 \notag \\
    &= \sigma_\ell \norm{S_{\ell-2} \Pi_{\ell-1} x_{\ell-2}}_2 \notag \\
    &\leq \sigma_\ell \left( \norm{S_{\ell-2} x_{\ell-2}}_2 + \norm{S_{\ell-2} \Pi_{\ell-1}^\bot x_{\ell-2}}_2\right) \notag \\
    &\leq \sigma_\ell \left( \norm{S_{\ell-2} x_{\ell-2}}_2 + \sigma_{\ell-1} \norm{S_{\ell-2} x_{\ell-2}}_2\right) = \sigma_\ell \, ( 1 + \sigma_{\ell-1}) \norm{S_{\ell-2} x_{\ell-2}}_2 \notag
\end{align}
to deduce that
\begin{align}
    \norm{S_{\ell-1} \Pi^\bot_{\ell} x_{\ell-1}}_2 
    &\leq \sigma_\ell \prod_{q = 1}^{\ell-1} ( 1 + \sigma_q) \| S_0 x_0\|_2 = \sigma_\ell \prod_{q = 1}^{\ell-1} ( 1 + \sigma_q) \|x\|_L. \notag
\end{align}

The end-to-end error $\| S \Pi^\bot x\|_2$ is controlled with a simple telescopic series argument. 
\begin{align}
    \| \Pi^\bot x\|_L =  \| S_0 \Pi^\bot x_0\|_2 
    &= \| S_0 x_0 - S_{c} x_{c} \|_2 \notag \\
    &\leq \| S_0 x_0 - S_{1} x_{1} \|_2 + \| S_1 x_1 - S_{2} x_{2} \|_2 + \ldots + \| S_{c-1} x_{c-1} - S_{c} x_{c} \|_2 \notag \\
    &=  \| S_0 \Pi_1^\bot x_{1} \|_2 + \| S_1 \Pi_1^\bot x_1 \|_2 + \ldots + \| S_{c-1} \Pi_c^\bot  x_{c-1} \|_2 \notag
\end{align}

Together, the above two results imply the desired bound.
%
\end{proof}

Therefore, to guarantee that in a multi-level scheme
\begin{align}
    \|\Pi^\bot B_0\|_L = \max_{b \in \Rbb^N} \frac{\|S \Pi^\bot B_0 \, b \|_2}{\|b\|_2}  \leq \epsilon, \notag
\end{align}
one needs to make sure that, for each level $\ell = 1, \ldots, c$, the following holds:
\begin{align}
    \frac{\|S_{\ell-1} \Pi_{\ell}^\bot x_{\ell-1}\|_2}{\|S_{\ell-1} x_{\ell-1}\|_2}  \leq \sigma_\ell, \quad \text{for all} \quad x_{\ell-1} = P_{\ell-1} \cdots P_{1} B_0\, b \notag
\end{align}
By the same argument used for the multi-level error, when $\ell=1$, we have that $\sigma_1 
= \| \Pi_1^\bot B_0\|_{L_{0}}$. For all other $\ell$, set $B_{\ell-1} = P_{\ell-1} \cdots P_{1} B_0$ and further let $(B_{\ell-1}^\top L_{\ell-1} B_{\ell-1})^{+\sfrac{1}{2}}$ be the pseudo-inverse of the matrix square-root of the $N\times N$ matrix $B_{\ell-1}^\top L_{\ell-1} B_{\ell-1}$. By the substitution $b = S_{\ell-1} B_{\ell-1} a$, the above can be rewritten as
\begin{align}
    \max_{b \in \Rbb^N} \frac{\|S_{\ell-1} \Pi_{\ell}^\bot B_{\ell-1} b\|_2}{\|S_{\ell-1} B_{\ell-1} b \|_2}
    = \max_{b \in \Rbb^N} \frac{\|S_{\ell-1} \Pi_{\ell}^\bot B_{\ell-1} (B_{\ell-1}^\top L_{\ell-1} B_{\ell-1})^{+\sfrac{1}{2}} b \|_2}{\|b \|_2}. \notag
\end{align} 
For $\ell>1$, therefore $\sigma_{\ell} = \|\Pi_{\ell}^\bot A_{\ell-1} \|_{L_{\ell-1}} $ with $A_{\ell-1} = B_{\ell-1} (B_{\ell-1}^\top L_{\ell-1} B_{\ell-1})^{+\sfrac{1}{2}}$.

\subsection{Proof of Proposition~\ref{proposition:decoupling}}


For notational simplicity in the context of this proof I drop level indices and assume that only a single coarsening level is used. Nevertheless, it should be stressed that this without loss of generality, as an identical argument holds for every level of the scheme.

Consider any $x$ and set $y = \Pi^\bot x$. Furthermore, define for each contraction set the (\textit{i}) \emph{internal} edge set $\E^{(r)} = \{ e_{ij} | v_i,v_j \in \V^{(r)}$ and $e_{ij} \in \E_{\ell-1} \}$ , and (\textit{ii}) the \emph{boundary} edge set $\partial\E^{(r)}$, such that if $e_{ij} \in \partial\E^{(r)}$ then $v_i \in \V^{(r)}$ and $v_j \notin \V^{(r)}$. It is true that 
\begin{align}
	\|\Pi^\bot x\|_{L}^2 
	&= \sum_{e_{ij} \in \E} w_{ij} (y(i) - y(j))^2 \notag \\ 
	&= \sum_{r = 1}^{n} \Bigg( \underbrace{\sum_{e_{ij} \in \E^{(r)}} w_{ij} (y(i) - y(j))^2}_{a_r}  + \frac{1}{2} \underbrace{\sum_{e_{ij} \in \partial\E^{(r)}} w_{ij} (y(i) - y(j))^2}_{b_r} \Bigg). \notag 
\end{align}
In the following, I will express $a_r$ and $b_r$ as a function of the vector $y_r = \Pi^\bot_{\V^{(r)}} x$.
Term $a_r$ is luckily independent of any other contraction set: 
$$
a_r = \sum_{e_{ij} \in \E^{(r)}} w_{ij} (y(i) - y(j) )^2 = \sum_{e_{ij} \in \E^{(r)}} w_{ij} (y_r(i) - y_r(j) )^2. 
$$ 
On the other hand, $b_r$ is smaller than 
\begin{align}
	b_r 
	&= \sum_{e_{ij} \in \partial\E^{(r)}} w_{ij} (y(i) - y(j))^2 
	\leq 2 \hspace{-3mm}\sum_{e_{ij} \in \partial\E^{(r)}} \hspace{-2mm} w_{ij} (y(i)-0)^2  + 2 \hspace{-3mm}\sum_{e_{ij} \in \partial\E^{(r)}} \hspace{-2mm} w_{ij} (0-y(j))^2. \notag
\end{align}
Distributing the second quantities, respectively, amongst the contraction sets that include said vertices, one gets
\begin{align}
	\|\Pi^\bot x\|_{L}^2 
	&\leq \sum_{r = 1}^{n} \Bigg( \sum_{e_{ij} \in \E^{(r)}} w_{ij} (y_r(i) - y_r(j) )^2 + 2 \sum_{e_{ij} \in \partial\E^{(r)}} w_{ij} (y(i)-0)^2 \Bigg) \notag \\
	&\hspace{0mm}= \sum_{r = 1}^{n} \Bigg( \sum_{e_{ij} \in \E^{(r)}} w_{ij} (y_r(i) - y_r(j) )^2 + \sum_{e_{ij} \in \partial\E^{(r)}} (2w_{ij}) (y_r(i)-y_r(j))^2 \Bigg) \notag \\
	&= \sum_{r = 1}^{n} \|y_r\|_{L_{\V^{(r)}}}^2
	= \sum_{\mathcal{C} \in \mathcal{P}} \|y_r\|_{L_{\mathcal{C}}}^2 . \notag
\end{align}
The second step above used the fact that $[\Pi^\bot](i) = 0$ for all $v_i \notin \mathcal{C}$.
A decoupled bound can then be obtained as follows:
\begin{align*}
\| \Pi^\bot A \|_{L}^2 
    = \max_{a \in \Rbb^{k-1}} \frac{\| S \Pi^\bot A \, a\|_2^2}{\|a\|_2^2} 
    &\leq \sum_{\mathcal{C} \in \mathcal{P}} \max_{a \in \Rbb^{k-1}} \frac{\| \Pi^\bot_{\mathcal{C}}\, A \, a\|_{L_{\mathcal{C}}}^2 }{\|a\|_2^2} = \sum_{\mathcal{C} \in \mathcal{P}} \| \Pi^\bot_{\mathcal{C}}\, A\|_{L_{\mathcal{C}}}^2
\end{align*}
The final inequality is derived by taking the square-root of the last equation.

\section{Complexity analysis}
\label{app:complexity}
%
The computational complexity of Algorithm~\ref{algorithm:multi-level} depends on the number of nodes $N$ and edges $M$ of $G$, the number of levels $c$, the subspace size $k$, as well as on how the families of candidate sets are formed. To derive worst-case bounds, I denote by $\Phi_\ell = \sum_{\C \in \mathcal{F}_{\ell}} |\C|$ the number of vertices in all candidate sets and by $\delta = \max_{\ell, \C \in \mathcal{F}_{\ell}} |\C|$ the cardinality of the maximum candidate set over all levels. Furthermore, I suppose that the per-level reduction ratio $r_\ell$ is a constant.

I start with some basic observations:
\begin{itemize}
    \item \emph{Computing $A_{0}, \ldots, A_{c-1}$ is possible in $\tilde{O}(ckM + k^2 N + ck^3)$ operations when $V=U_k$.}
    Each $A_{\ell-1}$ is computed once for each level. For $\ell=1$, one needs to approximate the first $k$ eigenpairs of $L$, which can be achieved in $\tilde{O}(kM)$ operations using inverse iteration as described by~\citet{vishnoi2013lx}. 
    For consecutive levels, forming matrix $B_{\ell-1}^\top L_{\ell-1} B_{\ell-1}$ takes $O(M_{\ell-1} k + N_{\ell-1} k^2)$ operations, whereas computing the pseudo-square-root $(B_{\ell-1}^\top L_{\ell-1} B_{\ell-1})^{+\sfrac{1}{2}}$ is possible in $O(k^3)$ operations. Summed up, the costs for all levels amount to $O( k \sum_{\ell=1}^{c} M_{\ell-1} + k^2 \sum_{\ell=2}^{c} N_{\ell-1} + c k^3) = O(ckM + k^2 N + ck^3) $, where I used the observation that $\sum_{\ell=2}^{c} N_{\ell-1} = O(N)$. 
    \item \emph{At each level, the cost function is evaluated at most $ \Phi_{\ell}$ times.}
    One starts by computing the cost of each candidate set in $\mathcal{F}_{\ell}$. Moreover, every $\C$ added to $\mathcal{P}_{\ell}$ causes the pruning of at most $\sum_{v_i \in \C} (\phi_i-1)$ other sets, where $\phi_i$ is the number of candidate sets that include $v_i$. Since $\mathcal{P}_{\ell}$ is a partitioning of $\mathcal{V}_{\ell-1}$, at most $\sum_{\C \in \mathcal{P}_{\ell}}\sum_{v_i \in \C} (\phi_i-1) \leq \sum_{v_i \in \mathcal{V}_{\ell-1}} \phi_i - |\mathcal{F}_{\ell}| = \Phi_\ell - |\mathcal{F}_{\ell}|$ cost re-evaluation are needed.

    \item \emph{Given $A_{\ell-1}$, each call of $\text{cost}_{\ell}(\C)$ requires $O(\min\{ k^2 \delta + k \delta^2,\ k\delta^2 + \delta^3\})$ operations.}
    The involved matrices themselves can be easily formed since, excluding all-zero rows and columns, both $L_{\C}$ and $\Pi_{\C}^\bot$ are $ |\C| \times |\C|$ matrices and one can safely restrict $A_{\ell-1}$ to be of size $|\C| \times k$ by deleting all rows that would have been multiplied by zero.
    Now, by definition, the incidence matrix $S_{\C}$ of $L_{\C}$ has at most $\delta$ columns and $2\delta$ rows (since one can bundle all boundary weights of a vertex in $\C$ in a single row). Depending on the relative size of $ k$  and $\delta$ the computation can be performed in two ways: 
    \begin{itemize}
    \item Either one forms the $k\times k$ matrix $A_{\ell-1}^\top \Pi_{\C}^\bot L_{\C} \Pi_{\C}^\bot \, A_{\ell-1}$ and approximate its spectral norm paying a total of $O(k^2 \delta + k \delta^2)$.      
    
    \item Otherwise, the $2d\times 2d$ matrix $S_{\C} \Pi_{\C}^\bot \, A_{\ell-1} A_{\ell-1}^\top \Pi_{\C}^\bot S_{\C}^\top $ is formed and its norm is computed at a combined cost of $O(\delta^2 k + \delta^3)$.
    \end{itemize}
    
    \item \emph{Maintaining $\mathcal{F}_{\ell}$ sorted incurs $O(\Phi_\ell \log{|\mathcal{F}_{\ell}|})$ cost.} Sorting $\mathcal{F}_{\ell}$ during initialization entails $O(|\mathcal{F}_{\ell}|\log{|\mathcal{F}_{\ell}|})$ operations. Inserting each $\C'$ into $\mathcal{F}_{\ell}$ (see step~\ref{algorithm:prunning}) can be done in $O(\log{|\mathcal{F}_{\ell}|})$ and, moreover, by the same argument used to bound the number of cost evaluations, at most $\Phi_\ell - |\mathcal{F}_{\ell}|$ such insertions can happen.
    
    \item \emph{Other operations carry negligible cost.} In particular, by implementing \textsf{marked} as a binary array, checking if a vertex is marked or not can be done in constant time.
\end{itemize}

Overall, using Algorithm~\ref{algorithm} and for $\mathbf{R}=\mathbf{U}_k$ one can coarsen a graph in $\tilde{O}(ckM + k^2N + ck^3 + \sum_{\ell = 1}^c  \Phi_{\ell} ( \min \{ k^2 \delta + k \delta^2,\ k\delta^2 + \delta^3\} + \log{|\mathcal{F}_{\ell}|}) )$ time, where the asymptotic notation hides poly-logarithmic factors.